\documentclass{article}

\usepackage{arxiv}
\usepackage{cite}
\usepackage{amsmath,amssymb,amsfonts,bm,amsthm}
\usepackage{mathtools}
\usepackage{algorithmic}
\usepackage{graphicx}
\usepackage{textcomp}
\usepackage{dsfont}
\usepackage{epstopdf}
\usepackage{orcidlink}
\usepackage{enumerate}
\usepackage{lscape}
\usepackage{multicol}
\usepackage{multirow}
\usepackage{booktabs}
\usepackage{framed}  
\usepackage{empheq}
\usepackage{graphics} 
\usepackage{epsfig} 
\usepackage{times} 
\usepackage{graphicx}
\usepackage{comment} 
\usepackage{apptools}
\usepackage{theoremref}
\usepackage{balance}

\usepackage{cite}
\usepackage{amsmath,amssymb,amsfonts}
\usepackage{nicematrix}
\usepackage{algorithmic}
\usepackage{graphicx}
\usepackage{algorithm,algorithmic}
\usepackage{soul}
\usepackage{xcolor}
\usepackage{tikz}
\usepackage{color}
\usepackage{graphicx}

\usepackage{subcaption}

\usepackage{textcomp}

\usepackage{optidef}

\usepackage{tikz, pgfplots}

\pgfplotsset{compat=1.18}
\usetikzlibrary{patterns, positioning, calc, backgrounds}

\definecolor{grey}{rgb}{0.9,0.9,0.9}

\usepackage{array}
\usepackage{eqparbox}
\usepackage{url}
\usepackage{relsize}
\usepackage{dsfont}

\newtheorem{theorem}{Theorem}
\newtheorem{lemma}[theorem]{Lemma}
\newtheorem{proposition}[theorem]{Proposition}

\newtheorem{definition}{Definition}

\newtheorem{remark}{Remark}

\newcommand{\bD}{{\mathbb{D}}}
\newcommand{\IR}{{\mathbb{R}}}
\newcommand{\CI}{{\mathbb{C}}}
\newcommand{\TP}{^\mathsf{T}}
\newcommand{\HP}{^\mathsf{H}}

\newcommand{\overbar}[1]{\mkern 1.5mu\overline{\mkern-1.5mu#1\mkern-1.5mu}\mkern 1.5mu}

\newcommand{\allone}[1]{\mathbf{1}_{#1}}
\newcommand{\rmin}{\rho_{\mathrm{min}}}

\DeclareMathOperator{\diag}{diag}
\DeclareMathOperator{\spec}{spec}
\DeclareMathOperator{\re}{Re}
\DeclareMathOperator{\im}{Im}

\DeclareMathOperator*{\argmin}{arg\,min}

\usepackage{letltxmacro}
\LetLtxMacro\orgvdots\vdots
\LetLtxMacro\orgddots\ddots

\makeatletter
\DeclareRobustCommand\vdots{%
  \mathpalette\@vdots{}%
}
\newcommand*{\@vdots}[2]{%
  \sbox0{$#1\cdotp\cdotp\cdotp\m@th$}%
  \sbox2{$#1.\m@th$}%
  \vbox{%
    \dimen@=\wd0 %
    \advance\dimen@ -3\ht2 %
    \kern.5\dimen@
    \dimen@=\wd2 %
    \advance\dimen@ -\ht2 %
    \dimen2=\wd0 %
    \advance\dimen2 -\dimen@
    \vbox to \dimen2{%
      \offinterlineskip
      \copy2 \vfill\copy2 \vfill\copy2 %
    }%
  }%
}
\DeclareRobustCommand\ddots{%
  \mathinner{%
    \mathpalette\@ddots{}%
    \mkern\thinmuskip
  }%
}
\newcommand*{\@ddots}[2]{%
  \sbox0{$#1\cdotp\cdotp\cdotp\m@th$}%
  \sbox2{$#1.\m@th$}%
  \vbox{%
    \dimen@=\wd0 %
    \advance\dimen@ -3\ht2 %
    \kern.5\dimen@
    \dimen@=\wd2 %
    \advance\dimen@ -\ht2 %
    \dimen2=\wd0 %
    \advance\dimen2 -\dimen@
    \vbox to \dimen2{%
      \offinterlineskip
      \hbox{$#1\mathpunct{.}\m@th$}%
      \vfill
      \hbox{$#1\mathpunct{\kern\wd2}\mathpunct{.}\m@th$}%
      \vfill
      \hbox{$#1\mathpunct{\kern\wd2}\mathpunct{\kern\wd2}\mathpunct{.}\m@th$}%
    }%
  }%
}
\makeatother

\begin{document}


\title{
Tannenbaum's gain-margin optimization\\
meets Polyak's heavy-ball algorithm\\[1ex]\large 
Dedicated to the memory of Allen R.~Tannenbaum and Boris T.~Polyak
\thanks{Supported in part by Hong Kong RGC under the project CityU 11203321, CityU 11213322, by the City University of Hong Kong under Project 9380054, and by the NSF/USA (ECCS-2347357), AFOSR/USA (FA9550-24-1-0278), and ARO/USA (W911NF-22-1-0292).}}

\author{Wuwei Wu\orcidlink{0000-0003-3754-8427}\thanks{W.~Wu and J.~Chen are with the Department of Electrical Engineering, City University of Hong Kong, Hong Kong SAR, China (e-mail: \href{mailto:w.wu@my.cityu.edu.hk}{w.wu@my.cityu.edu.hk}; \href{mailto:jichen@cityu.edu.hk}{jichen@cityu.edu.hk}).},
Jie Chen\orcidlink{0000-0003-3171-1887}\footnotemark[2],
Mihailo R.\ Jovanovi\'c\orcidlink{0000-0002-4181-2924}\thanks{M.\ R.\ Jovanovi\'c is with the Ming Hsieh Department of Electrical and Computer Engineering, University of Southern California, Los Angeles, CA 90089 USA (e-mail: \href{mailto:mihailo@usc.edu}{mihailo@usc.edu}).},
and Tryphon T.\ Georgiou\orcidlink{0000-0003-0012-5447}\thanks{T.\ T.\ Georgiou is with the Department of Mechanical and Aerospace Engineering, University of California, Irvine, CA 92697 USA (e-mail: \href{mailto:tryphon@uci.edu}{tryphon@uci.edu}).}}
\maketitle

\begin{abstract}
This paper highlights an apparent, yet relatively unknown link between algorithm design in optimization theory and controller synthesis in robust control. Specifically, quadratic optimization can be recast as a regulation problem within the framework of $\mathcal{H}_\infty$ control.
From this vantage point, the optimality of Polyak's fastest heavy-ball algorithm can be ascertained as a solution to a gain margin optimization problem. The approach is independent of Polyak's original and brilliant argument, and relies on foundational work by Tannenbaum, who introduced and solved gain margin optimization via Nevanlinna--Pick interpolation theory. The link between first-order optimization methods and robust control sheds new light on the limits of algorithmic performance of such methods, and suggests a framework where similar computational tasks can be systematically studied and algorithms optimized. In particular, it raises the question as to whether periodically scheduled algorithms can achieve faster rates for quadratic optimization, in a manner analogous to periodic control that extends the gain margin beyond that of time-invariant control. This turns out not to be the case, due to the analytic obstruction of a transmission zero that is inherent in causal schemes. Interestingly, this obstruction can be removed with implicit algorithms, cast as feedback regulation problems with causal, but not strictly causal dynamics, thereby devoid of the transmission zero at infinity and able to achieve superior convergence rates.
\end{abstract}
%

\vspace*{.15in}
%


\section{Introduction}

The seeds of our subject, to view optimization methods within the prism of the theory of dynamical systems, have a long history~\cite{tsypkin1971adaptation, tsypkin1973foundations, krasnoselʹskii1989positive, helmke_optimization_1994}. The systematic exploration of the control toolbox for algorithm analysis and design has received renewed impetus in recent years~\cite{bhaya2006control, kashima2007system,lessard2016analysis, hu2017dissipativity, fazlyab_analysis_2018,  muehlebach_dynamical_2019, scherer2021convex, Ugrinovskii2022, ross_generating_2023, ugrinovskii2023robust}. The theme of our paper builds on a link between the celebrated gain margin problem in the modern theory of robust control, due to Allen Tannenbaum, and the foundational heavy-ball algorithm of Boris Polyak, as pointed out by Ugrinovskii et al.~\cite{ugrinovskii2023robust}, and by authors of this paper in~\cite{zhang2024frequency}.

Our exposition begins with a concise account of the gain margin problem and the significance of {\em transmission zeros}, as analytic obstacles that limit performance.
First-order optimization methods are treated next, and are recast as tracking problems, where optimization algorithms serve as control laws to effect feedback regulation. When the objective function is quadratic, the convergence rate is tied to the gain margin that feedback regulation can achieve; we then highlight the fact that optimality of the heavy-ball algorithm can be deduced from the theory of the gain margin problem.

The foundational contributions in modern robust control of the early 1980's were followed by the development of periodic and sampled-data $\mathcal{H}_\infty$ methods. In this, a lifting isomorphism translates time-varying periodic dynamics into time-invariant ones with a special structure.
One of the first observations was that transmission zeros were no longer an obstacle
\cite{khargonekar1985robust,kabamba1987control,Lee1987,francis1988stability},
and thereby, that the achievable gain margin can be  improved; in fact, arbitrarily large gain margin can be achieved.
A succinct exposition of this theory will be provided, and it will be shown that when first-order optimization methods are implemented in a strictly causal
manner, the analytic obstacle due to ``causality zeros'' cannot be removed. Thereby, time-varying schemes cannot beat Polyak's heavy-ball performance.

The insights gained
suggest that faster algorithms need to circumvent the strict causality constraint.
This is indeed the case in
\emph{implicit} iterative first-order optimization schemes. It will be shown that these can likewise be cast as feedback regulation. The enabling difference is that now dynamics are causal but not strictly causal.
As a result, the analytic constraint from the zero at infinity disappears. Superior rates of implicit schemes are thus explained within the framework of gain margin optimization theory. This link provides an approach to optimize implicit optimization schemes, which offer advantages in ill-conditioned problems.
With the help of the circle criterion, the results are extended to algorithms for optimization of not-necessarily quadratic functions.

Our development
enables the analysis and the design of gradient-based algorithms as control problems, which in turn can be tackled by employing rich techniques found in the theory of robust control, based on frequency-domain input-output maps. Particular tools include the Nevanlinna--Pick interpolation method and stability theory of Lur'e systems. These facilitate the development of analytical solutions to analysis and design of first-order optimization algorithms.

The link between the two themes,
provides a framework to study algorithmic performance and robustness.
A case in point arises in applications where the exact value of the gradient is not fully available, e.g., when the objective function is obtained using simulations or noisy measurements. Such instances necessitate the study of noise amplification in optimization algorithms~\cite{mohrazjovTAC21,mohrazjovTAC24}.
This link is also suited for studying distributed optimization, as pursued in~\cite{zhang2024frequency} for special cases of time-invariant and strictly causal schemes.

The paper aims to have an expository value as well, since some of the background theory is not well known. To this end, we begin with a brief self-contained expository on the gain margin problem in Section~\ref{sec:gainmargin}. First-order optimization methods are introduced in Section~\ref{sec:firstorder} and are cast as feedback systems and the optimal choice of parameters as gain margin optimization.
Section~\ref{sec:periodic} discusses periodic control and explores the possibility of periodic schedules in carrying out the computations of first-order methods in a periodic manner, in hopes of alleviating analytic constraints that limit achievable rates; it is shown that this is not possible since the analytic obstacle arises from transmission zeros at infinity that cannot be removed,thus precluding improvements in convergence rates of algorithms including those for certain online optimization tasks.
This circle of ideas segues into Sections~\ref{sec:implicit_section} and \ref{sec:nonquadratic}, where it is explained how implicit methods, cast in a similar manner as feedback systems and implementable via proximal operators, avoid the analytic obstacle at infinity that is due to causality, and thereby achieve superior convergence rates.
The paper concludes with an epilogue, Section~\ref{sec:epilogue}, and a dedication to Allen Tannenbaum and Boris Polyak.
It is our hope that the developments in this paper, born in the confluence of their pioneering contributions, will offer a fruitful and potentially consequential re-thinking of optimization methods.

\subsection*{Notation and Preliminaries}
For a matrix $X$,  $X\TP$ denotes the transpose, $X\HP$ the conjugate transpose, and $\overbar{\sigma}(X)$ the maximum singular value. For symmetric matrices $X$ and $Y$ of equal dimension, $X(\succeq)\succ Y$ indicates that $X-Y$ is positive (semi-)definite. The spectrum of a square matrix $X$ is denoted by $\spec(X)$.
The Kronecker product is denoted by $\otimes$.
We denote the $d$-dimensional Euclidean space by $\IR^d$, in which the inner product of vectors $x$ and $y$ is denoted by $\left< x,y\right> \coloneqq x\TP y$ and the Euclidean norm of $x$ by $\left\|x\right\|$.
The identity matrix in $\IR^{n\times n}$ is denoted by $I_n$, the vector of ones in $\IR^d$ by $\mathbf{1}_d$, and the zero vector in $\IR^d$ by $\mathbf{0}_d$; subscripts are omitted when it is clear from the context.

The open unit disc is denoted by $\bD \coloneqq \{z\in\CI \colon |z|< 1\}$,
its boundary by $\partial\mathbb{D}\coloneqq\{z\in\mathbb{C}\colon|z|=1\}$, and its complement in $\overbar{\CI}\coloneqq\CI\cup\{\infty\}$ by $\bD^c \coloneqq \{z\in \overbar{\CI} \colon |z|\geq 1\}$.
The complex conjugate of $z\in\CI$ is denoted by $\overbar{z}$.
We consider finite-dimensional discrete-time linear time-invariant (LTI) dynamical systems. They are represented by their corresponding (rational) transfer function $G(z)=\widehat{g}(z)\coloneqq\sum^\infty_{t=0} g[t]z^{-t} $, the $\mathcal{Z}$-transform of their impulse response ${\{g[t]\}}^\infty_{t=0}$. A system is
\emph{stable} if its transfer function has no pole in ${{\bD}}^c$, i.e., it is analytic in $\bD^c\setminus\partial\bD$ and bounded on $\partial\bD$. The space that comprises all stable rational transfer functions is denoted by $\mathcal{RH}_{\infty}$.

\section{The gain margin problem and the dawn of robust control}\label{sec:gainmargin}

The stability of the standard feedback interconnection, displayed in Fig.~\ref{fig:feedback}, where $P$ signifies a given ``plant'' and $C$  the ``controller'' that often needs to be designed,
is equivalent to the stability of all ``closed-loop'' transfer functions, from the external inputs $u_0$, $y_0$ to signals generated inside the feedback loop, $u_1$, $y_1$, $u_2$, $y_2$ (Stability of the transfer functions from $u_0$, $y_0$ to $u_1$, $y_1$ suffices).
In control design, of special importance is the transfer function from $y_0$ to $y_1$, i.e., from ``measurement noise'' to ``system output,'' namely,
\begin{equation*}
T(z) \coloneqq P(z)C(z) {\bigl(I+P(z)C(z)\bigr)}^{-1},
\end{equation*}
and is referred to as
the \emph{complementary sensitivity} function.

\begin{figure}[tb]
\centering
\begin{tikzpicture}[>=latex, line width=1pt]
    \tikzstyle{component} = [draw, rectangle, very thick, inner sep=1ex, minimum height=1em]
     \tikzstyle{add} = [draw, circle, inner sep=2pt]
    
          \node[component] (c) {$C(z)$};
     \node[component, above = 3em of c] (p) {$P(z)$};

    \node(a1) [add, right=6em of c]  {};
      \node(a2) [add, left=6em of p] {};

    \draw[->] (a2) --node[above]{$u_1[t]$} (p);
    
    \draw[->] (p) -|node[above,pos=0.23]{$y_1[t]$} node[left,pos=0.9]{$+$} (a1);
    
    \draw[->] (a1) --node[above]{$y_2[t]$} (c);
    
    \draw[->] (c) -| node[above,pos=0.23]{$u_2[t]$} node[pos=0.9,right]{$-$} (a2);
    
    \draw[<-] (a2) -- node[pos=0.8,above]{$u_0[t]$} ++(-4em,0);
    
    \draw[<-] (a1) --node[pos=0.8,above]{$y_0[t]$} ++(4em,0);
       \end{tikzpicture}
\caption{Feedback control system.}
\label{fig:feedback}
\end{figure}
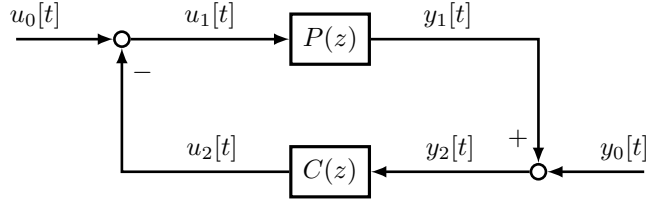

When $P$ and $C$ are scalar transfer functions, i.e., when the corresponding systems have a single input and a single output, the conditions for stability of the feedback interconnection reduce to ensuring that
the complementary sensitivity
is stable and that no unstable pole-zero cancellation takes place between $P(z)$ and $C(z)$.
The problem to design a control law by selecting a suitable transfer function $C(z)$ then reduces to determining the corresponding complementary sensitivity instead, since $C(z)$ can be readily computed from the known $P(z)$.
The benefit of doing so is that the design problem reduces to determining $T(z)$ that is
analytic and bounded in  $\bD^c$,
and satisfies the additional conditions
\begin{equation}\label{eq:intcond}
T(p_i)=1 \;\text{and}\; T(s_j)=0,
\end{equation}
where $p_i$, $s_j$
represent the poles and zeros of $P(z)$ in $\bD^c$, respectively. Additional control design specifications are often conveniently expressed in terms of the range of values of $T(z)$.


Allen Tannenbaum recognized the close linkage between control design objectives and analytic constraints, and brought center stage conformal transformations for expressing control objectives in the formalism of analytic interpolation theory (see~\cite[Chapter 11]{doyle2013feedback} and~\cite{khargonekar1985non}).
This led to the celebrated optimal gain margin problem~\cite{Tannenbaum1980,tannenbaum1982modified},~\cite[Chapter VIII]{tannenbaum2006invariance}:

\vspace*{.5ex}

\noindent
\colorbox{grey}{\begin{minipage}[]{.98\textwidth}
\noindent {\em Tannenbaum's gain margin optimization:}
    Let $P_0(z)$ be the transfer function of a scalar dynamical system. Determine the transfer function $C(z)$ for a feedback control law that stabilizes $P_0$ along with the family
\begin{equation}\label{eq:gain_uncertain}
    \mathcal{P}\coloneqq\left\{ k P_0 \colon k_1\leq k \leq k_2\right\},
\end{equation}
(where we assume that $1\in[k_1,k_2]$),
and achieves a maximal margin $20 \log(k_2/k_1)$ [dB] for the (uncertain) gain factor $k$.
   \end{minipage}}

\vspace*{.5ex}

The plant model $P_0$ is considered the nominal, and a controller is sought to achieve a maximal deviation of $k$ from the nominal value of $1$. Thus, it is assumed that $0<k_1<1<k_2$. Further, one may take $1=\sqrt{k_1k_2}$ to be the geometric mean, so that the feedback system remains stable for an equal deviation 
above and below the nominal value $k=1$.

\subsection{Solution of the gain margin problem}


\begin{figure}[tb]
\centering
\begin{tikzpicture}[>=latex, line width=1pt]
    \tikzstyle{component} = [draw, rectangle, very thick, inner sep=1ex, minimum height=1em]
     \tikzstyle{add} = [draw, circle, inner sep=2pt]

    \begin{scope}[local bounding box=scc]
     \node[component] (p) {$P_0(z)$};
     \node[right = 4em of p.east] (k0) {};
    \node[component, below right=3em and 1em of p.east] (c) {$C(z)$};
        \node(a1) [add, right=4.5em of c]  {};
      \node(a2) [add, left=2.3em of p] {};
      \node(a3) [add, right=2em of k0] {};

    \draw[->] (a3) -| node[left,pos=0.9]{$+$} (a1);
    
    \draw[->] (a1) -- (c);
    
    \draw[->] (c) -| node[pos=0.9,right]{$-$} (a2);

        \draw[->] (a2) -- (p);

    \draw[->] (p) -- (a3);
     
               \begin{pgfonlayer}{background}
                   \path[fill=gray!10, draw=black, dashed, very thick, rounded corners] ([xshift=-1.3em, yshift=1ex] scc.north west) rectangle ([xshift=1.3em, yshift=-1ex] scc.south east);
               \end{pgfonlayer}
    \end{scope}
          
    \node at([xshift=-1.6em, yshift=0.5em] scc.south west)[inner sep=0] {$-T(z)$};

     \node[component, above = 1.5em of k0.north] (k) {$k-1$};

    \draw[->] ([xshift=1.5em]p.east) |- (k);
    \draw[->] (k) -| (a3);

    \draw[<-,dotted] (a2) -- node[pos=0.8,above]{$u_0[t]$} ++(-4em,0);
    
    \draw[<-,dotted] (a1) --node[pos=0.8,above]{$y_0[t]$} ++(4em,0);
       \end{tikzpicture}
\caption{Feedback control system under gain uncertainty.}
\label{fig:gain_uncertain}
\end{figure}
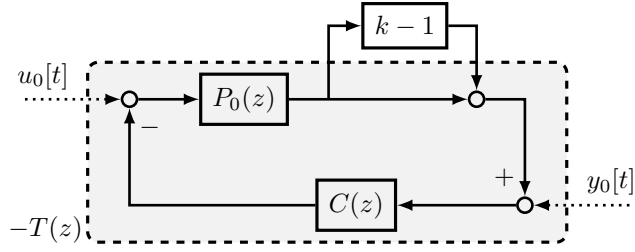

 Fig.~\ref{fig:gain_uncertain} is a re-drawing of Fig.~\ref{fig:feedback}
 that exemplifies the uncertain family as generated by an added forward path with gain $k-1$.
 The transfer function from the output of this added component, represented by the box with gain $k-1$, to its input, can be computed as being the \emph{negative of} (as highlighted in Fig.~\ref{fig:gain_uncertain})
 \begin{equation}\label{eq:T0}
    T(z)=\frac{P_0(z)C(z)}{1+P_0(z)C(z)}.
\end{equation}
 This is the nominal complementary sensitivity.
Hence, to guarantee robust stability, other than the stability of $T(z)$ (i.e., the standard notion of nominal stability of the system in Fig.~\ref{fig:gain_uncertain}),
it is necessary and sufficient that
its closed-loop characteristic function contains no zero, that is,
$1+(k-1)T(z)\neq 0$, or equivalently, $T(z)\neq \frac{-1}{k-1}$ for all $z\in\bD^c$ and $k\in[k_1,k_2]$.
%
As such, the permissible range of values of $\zeta=T(z)$ (for $z\in\mathbb D^c$) must exclude values in
\begin{equation}\label{eq:Sk1k2}
    \mathbb{S}_{k_1,k_2} \coloneqq \Bigl\{\zeta\in\IR \colon \zeta \leq \frac{-1}{k_2-1} \;\text{or}\; \zeta \geq \frac{-1}{k_1-1} \Bigr\}.
\end{equation}
It is easy to see that the composition of three conformal maps,
\begin{alignat*}{3}
   \phi_{k_1,k_2} \;&\colon  \zeta &&\mapsto\;
w={\frac{1+(k_1-1)\zeta}{1+(k_2-1)\zeta}},\\
 \sqrt{\cdot} \;&\colon w &&\mapsto\;
v=\sqrt{w},\\
 \psi \;&\colon v &&\mapsto\;
u=\frac{1-v}{1+v}
\end{alignat*}
map successively the range of permissible values, namely, the complement set $\mathbb{S}_{k_1,k_2}^c$, first into the complement  ${\left\{z \in \mathbb{C} \colon \re(z) \leq 0, \im(z) =0 \right\}}^c$ of the negative real line, then into the right half complex plane, and finally into $\mathbb D$.

This composition map
\begin{equation}\label{eq:Phi}
u=\mathbf{\Phi}_{k_1,k_2}(\zeta) := \psi\circ \sqrt{\phi_{k_1,k_2}(\zeta)},
\end{equation}
giving
\begin{equation*}
    u=
\frac{\sqrt{1+(k_2-1)\zeta} - \sqrt{1+(k_1-1)\zeta}}{\sqrt{1+(k_2-1)\zeta} + \sqrt{1+(k_1-1)\zeta}},
\end{equation*}
can be readily inverted so that
\begin{equation}\label{eq:inv_conformal}
        \zeta = \mathbf{\Phi}_{k_1,k_2}^{-1}(u)
        = {\left( \tfrac{k_2-k_1}{4} \left( u + {u}^{-1}\right) - \tfrac{k_2+k_1}{2} +1 \right)\!}^{-1}.
\end{equation}

To recap, the sought controller can be obtained from a $T(z)$, analytic in $\mathbb D^c$,  that satisfies \eqref{eq:intcond} with range that excludes $\mathbb{S}_{k_1,k_2}$. It readily follows that
\begin{equation}\label{eq:bfT}
    \mathbf{T}(z):= \mathbf{\Phi}_{k_1,k_2}\circ T(z)
\end{equation}
must also be analytic in $\mathbb D^c$ with range in $\mathbb D$, and satisfy
\begin{subequations}\label{eq:intcond1}
\begin{align}
\mathbf T(p_i) &= \mathbf{\Phi}_{k_1,k_2}(1)=g(k_1,k_2)\coloneqq\frac{\sqrt{k_2/k_1}-1}{\sqrt{k_2/k_1}+1},\\
\mathbf T(s_j) &= \mathbf{\Phi}_{k_1,k_2}(0)=0.
\end{align}
\end{subequations}
To construct such a $\mathbf T(z)$ is the classical Nevanlinna--Pick interpolation problem. We briefly review the essentials before we complete the expository on Tannenbaum's gain margin problem.

\subsection{Rudiments of analytic interpolation theory}\label{sec:interpolation}

Two pillars of the classical analytic interpolation theory are the maximum modulus theorem and the Mobius transformation. 
The first states that unless the function is constant throughout, the maximum modulus can only be achieved on the boundary.
The second pillar allows a systematic way to reduce the number of interpolation constraints via a succession of conformal transformations in the Nevanlinna--Pick algorithm. Both can be encapsulated in Schur's lemma:

\newcommand{\bT}{\mathbf{T}}
\emph{A function $\bT(z)$ is analytic in $\mathbb D^c$ (which can be replaced by any other conformally equivalent domain), satisfies $\bT(z_0)=w_0$ and maps $\mathbb D^c$ to $\mathbb D$ if and only if
\begin{equation*}
    \bT_{\rm next}(z):= \frac{1-\overbar{z_0} z}{z-z_0}\frac{\bT(z)-w_0}{1-\overbar{w_0}\overbar{\bT(z)}}
\end{equation*}
is analytic in $\mathbb D^c$ and maps $\mathbb D^c$ to $\mathbb D$.}

Any additional interpolation conditions on $\bT(z)$ are inherited by $\bT_{\rm next}(z)$, revised as per the conformal map. For example, if it is required that $\bT(z_i)=w_i$, then $\bT_{\rm next}(z_i)$ can be readily computed from their conformal relation shown above. The problem to ascertain whether such a function exists, and to construct one such, then becomes simpler at each successive step by removal of one interpolation condition until only one condition is left, in which case a constant function can be chosen. The process in reverse can be used to construct $\bT(z)$. This is the \emph{Nevanlinna--Pick algorithm} \cite[Section 9.3]{doyle2013feedback}.

Thus, successive application of Schur's lemma provides a solution to Tannenbaum's problem in the most general case when $P_0(z)$ has multiple poles and zeros in $\mathbb D^c$. For the problem at hand, where we want to assess the optimality of the heavy-ball method in the next section, we only need to consider nominal plants $P_0(z)$ that have unstable poles at $p_{i}$, $i=1,\dots,n_{p}$, and a non-minimum phase zero at $s=\infty$. In this case, and for the conditions in \eqref{eq:intcond1}, the gain margin problem has a solution if and only if
\begin{equation}\label{eq:unmistable}
 \frac{\sqrt{k_2/k_1}-1}{\sqrt{k_2/k_1}+1} =  g(k_1,k_2) < \prod_{i=1}^{n_p} \frac{1}{\left| p_i \right|} ,
\end{equation}
and is achieved by taking
\begin{equation}\label{eq:unmistable_T}
    \bT(z) =  g(k_1,k_2)
    \frac{\prod_{i=1}^{n_p} \overbar{p_{i}} ( z -  p_{i} ) - \prod_{i=1}^{n_p} ( \overbar{p_{i}} z -  1 )}{\prod_{i=1}^{n_p} (p_{i}^{-1} z - 1 ) - \prod_{i=1}^{n_p} ( \overbar{p_{i}} z -  1 )}
\end{equation}
and constructing the corresponding controller $C(z)$ via \eqref{eq:inv_conformal} and \eqref{eq:T0}.
For future reference, we also note that if $P_0(z)$ has no non-minimum phase zeros, even at $\infty$, i.e., it is causal but not strictly causal, then the achievable gain margin $k_2/k_1$ can be made arbitrarily large.

\section{First-order optimization: from Cauchy to Polyak and Nesterov}\label{sec:firstorder}

The significance of the problem to
\begin{equation}\label{eq:opt_problem}
	\mathrm{minimize} \left\{f(x)\in\mathbb R \colon x\in\IR^d\right\}
\end{equation}
is fundamental across engineering disciplines. To address such problems, there is a plethora of methods and literature, with the gradient descent method that can be traced back to Cauchy and the Newton method as perhaps the two most widely known.
However, when one considers minimization problems in very high dimensions, Newton's method is not an option and gradient descent is often plagued by very slow convergence.
To this end, accelerated methods have been devised, most notably,
Polyak's heavy-ball method for quadratic problems \cite{Polyak1964} and Nesterov's method for convex functions that followed shortly afterward \cite{nesterov2018lectures}. We herein focus on Polyak's method.

 In the sequel, we initially consider general functions $f(x)$ to show how to formulate an optimization algorithm as a feedback system, and we then specialize to
quadratic  functions in the form of
\begin{equation}\label{eq:quadratic_cost}
    f(x) = \frac{1}{2} {x}\TP Q x - {q}\TP x,
\end{equation}
where $Q\in\IR^{d\times d}$ is  symmetric and positive definite, with spectrum $\spec(Q)\subset [\mu,\ell]$.
We denote the class of such functions by $\mathcal{Q}_{\mu,\ell}$, which is a subset of $\mathcal{F}_{\mu,\ell}$ constituted by functions that are both $\mu$-strongly convex and $\ell$-Lipschitz smooth as defined below.
\begin{definition}\label{def:convex}
    A function $f\colon\IR^d\to\IR$ is said to be \emph{$\mu$-strongly convex} for some $\mu>0$ if the function $f(x)-\frac{\mu}{2} \|x\|^2$ is convex,
 and is said to be \emph{$\ell$-Lipschitz smooth} if its gradient $\nabla f$ is $\ell$-Lipschitz continuous, i.e., $\|\nabla f(x)-\nabla f(y)\|\leq \ell\|x-y\|$ for all $x,y\in\IR^d$.
\end{definition}

If $f$ is $\ell$-Lipschitz smooth, then $\frac{\ell}{2} \|x\|^2 - f(x)$ is convex \cite[Lemma 1.2.3]{nesterov2018lectures}. Hence, if $\mathcal{F}_{\mu,\ell}$ is not empty, then $\mu\leq\ell$.
Focusing on the optimization of quadratic functions $f\in\mathcal{Q}_{\mu,\ell}$, Polyak introduced the following method  \cite{Polyak1964}.

\vspace{.5ex}


\noindent
\colorbox{grey}{\begin{minipage}[]{.98\textwidth}
{\em Polyak's heavy-ball method:} It consists of the iterative algorithm
\begin{equation}\label{eq:heavy_ball}
   x[t+1] = x[t] + {\left(\frac{\sqrt{\ell}-\sqrt{\mu}}{\sqrt{\ell}+\sqrt{\mu}}\right)\!}^2 \left(x[t]-x[t-1]\right) - \frac{4}{{\bigl( \sqrt{\ell} + \sqrt{\mu} \bigr)\!}^2}  \nabla f(x[t]),
\end{equation}
shown by Polyak to achieve convergence rate 
\begin{equation}\label{eq:rho_min}
    \rmin \coloneqq
    \frac{\sqrt{\kappa}-1}{\sqrt{\kappa}+1},
\end{equation}
with $\kappa={\ell}/{\mu}$ the \emph{condition number} of $f(x)$. It was shown by Nesterov \cite[Theorem 2.1.13]{nesterov2018lectures} that this rate is in fact the best possible (see also \cite{ugrinovskii2023robust} and the discussion in Section \ref{sec:periodic_limitation}).
\end{minipage}}

\vspace{.5ex}

We formally define convergence rate below, linked to the zero-input response of a stable feedback system.
\begin{definition}\label{def:converge}
	A sequence ${\left\{{e}[t]\right\}}_{t=0}^{\infty}$ is said to {\em converge to zero at a rate} $\rho\in \left(0,1\right)$ if
 \begin{equation}\label{eq:rho_def}
 \rho = \inf \{ \gamma \in (0,1) \colon \gamma^{-t} {e}[t]\to\mathbf{0} \text{ as }  t\to\infty\},
 \end{equation}
and then, $\rho$ is referred to as the \emph{convergence rate of} ${\left\{\mathbf{e}[t]\right\}}_{t=0}^{\infty}$. Moreover, if the zero-input response of a dynamical system converges to zero at a rate $\rho$, the system is said to be \emph{$\rho$-stable}.
\end{definition}

\begin{remark}\label{rmk:converge}
    The convergence rate $\rho$, referred to as \emph{linear convergence rate} in numerical analysis, can be alternatively defined by any of the following equivalent expressions,
    \begin{subequations}
            \begin{align}
       \rho &=\inf\{\gamma\in (0,1) \colon \|e[t]\| \leq c \gamma^t \|e[0]\| \text{ for some }c\} \label{eq:rho1}\\
        &= \limsup_{t\to \infty} \|e[t+1]\| / \|e[t]\| \label{eq:rho2}\\
        &= \limsup_{t\to \infty} \|e[t]\|^{1/t} \label{eq:rho3}.
    \end{align}
    \end{subequations}
These are commonplace in the optimization literature (e.g., \cite{Polyak1964, goujaud2023heavyball}), and recently, in control-theoretic studies of optimization algorithms as well (e.g., \cite{lessard2016analysis, scherer2021convex, hu2017dissipativity, Ugrinovskii2022, ugrinovskii2023robust}).
Our choice of \eqref{eq:rho_def} as the defining expression is meant to highlight the fact that ${\left\{{e}[t]\right\}}_{t=0}^{\infty}$ converges to zero at a rate $\rho$ if and only if its $\mathcal{Z}$-transform is analytic in ${\left\{ z \in \mathbb{C} \colon |z| > \rho \right\}}$.
It should be noted that $\rho$ quantifies \emph{asymptotic convergence} and does not reflect transient behavior.
Thus, it is possible that an algorithm may have poor performance due to transient effects that are not reflected in the corresponding value for $\rho$.
    \null\hfill$\Diamond$
\end{remark}

The parallel between \eqref{eq:rho_min} and \eqref{eq:unmistable} is unmistakable. Next, we explain how optimization schemes can be seen as control problems, so as to highlight the parallels, and reconvene to re-establish the optimality of Polyak's method with time-invariant as well as periodic computing schemes. In passing, we note that the heavy-ball method has been studied for various classes of functions beyond quadratic functions (see, e.g., \cite{Ugrinovskii2022, goujaud2023heavyball}).

\subsection{Optimization methods as control systems}\label{sec:opt_control}

First-order optimization methods utilize solely gradient information on the function $f(x)$ to be minimized. The simplest case of gradient descent takes the form
\begin{equation}\label{eq:example}
x[t+1]=x[t]-\alpha \nabla f(x[t]),
\end{equation}
where $\alpha$ is a step size.
In general, assuming that $x^\star$ is the solution to the optimization problem \eqref{eq:opt_problem} and letting
$e[t] \coloneqq x^\star-x[t]$
be the optimization error, the gradient descent iteration in \eqref{eq:example} can be seen as the response of a dynamical system with transfer function
\begin{equation*}
    G(z)=\frac{\alpha}{z-1}
\end{equation*}
 and input $-\nabla f(x[t])$. The scheme is drawn in Fig.~\ref{fig:alg_feedback} so as to highlight the nature of the problem as a \emph{tracking problem}. This can be seen by rewriting the equation \eqref{eq:example} in the $\mathcal{Z}$-domain error form
\begin{equation*}
    \widehat{e}(z) = \frac{z}{z-1} x^\star -G(z) \widehat{\Delta_f(e)} (z),
\end{equation*}
where $\widehat{\Delta_f(e)} (z)$ is the $\mathcal{Z}$-transform of the sequence ${\{ \Delta_f(e[t]) \}}_{t=0}^\infty$, and
\begin{equation}\label{eq:Delta}
    \Delta_f(e[t]) \coloneqq - \nabla f(x^\star-e[t]) = - \nabla f(x[t])
\end{equation}
is memoryless with $\Delta_f(\mathbf{0})=\mathbf{0}$ and in many cases nonlinear.
Assuming that each coordinate of $-\nabla f(x[t])$ is treated equally, (i.e., the algorithm is dimension-free), the transfer function takes the form $G(z)I_d$, with $G(z)$ a scalar quantity.
The dynamics in $G(z)$ can be quite general, suitably enhanced to retain past information so as to effect suitable processing and accelerate convergence. Clearly, the feedback scheme as drawn always relies on first-order derivative information. As a special case, it includes the structure of the heavy-ball method, as we will explain shortly.

\begin{figure}[tb]
\centering
\begin{tikzpicture}[>=latex, line width=1pt]




    \tikzstyle{component} = [draw, rectangle, very thick, inner sep=1ex, minimum height=1em]
     \tikzstyle{add} = [draw, circle, inner sep=3pt]

     \node[component] (d) {$\Delta_f(\cdot)$};

      \node(a1) [add, left=3.6em of d] {};

      \node[component, right = 6.5em of d] (m) {${G}(z)I_d$};

    \draw[->] (a1) --node[above,pos=0.4]{$e[t]$} (d);
    \draw[->] (d) --node[above,pos=0.47]{$-\nabla f(x[t])$} (m);
    \draw[->] (m.east) -- node[above,pos=0.5]{$x[t]$} ++(4em,0);
    \draw[<-] (a1) -- node[pos=0.8,above]{$x^\star$} ++(-3em,0);
    \draw[->] ([xshift=2em]m.east) -- ++(0,-4em) -| node[pos=0.88,right=-0.5ex]{$-$} (a1);

       \end{tikzpicture}
\caption{First-order optimization algorithm as a feedback system.}
\label{fig:alg_feedback}
\end{figure}
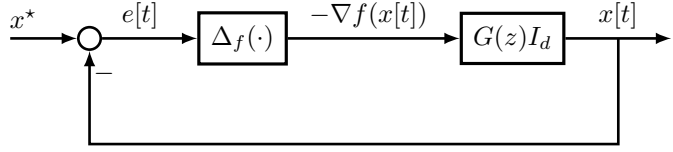

Inspecting the scheme in Fig.~\ref{fig:alg_feedback}, the problem of designing an algorithm that steers $x[t]$ towards $x^\star$ is equivalent to synthesizing a controller such that the output $x[t]$ tracks the reference step signal with \emph{unknown} magnitude $x^\star$.
It follows from the well-known \emph{internal model principle}, a classical result in control theory \cite{francis1975internal, francis_internal_1976}, that for ${x}[t]$ to track $x^\star$, $G(z)$ must contain the mode of the reference signal; that is, in the present case $G(z)$ must have an accumulator (i.e., a pole at $1$ that provides the analog of an integrator in discrete time).

The response of the feedback system to a step signal $x^\star$ is identical to the response due to an initial state of $1/(z-1)$ without external excitation.
Thus, the convergence of the optimization algorithm is equivalent to the asymptotic stability of the feedback system without exogenous input, starting from a suitable initial condition. We summarize as follows.

\begin{proposition}\label{prop:LTI}
     Let $f$ be differentiable and let the minimizer $x^\star\in\IR^d$ be a unique stationary point of $f$.
     Then, the optimization algorithm characterized by transfer function $G(z)$ converges to the unknown minimizer $x^\star$ if and only if the feedback system in Fig.~\ref{fig:alg_feedback} is asymptotically stable without exogenous input and $G(z)$ has an accumulator.
\end{proposition}

\begin{proof}
It follows by internal model principle \cite{francis_internal_1976}. \end{proof}

We emphasize that $G(z)$ must be strictly causal in order to avoid circular dependence, i.e., using $\nabla f(x[t+1])$ in the computation of $x[t+1]$. (We note in passing that the circular dependence can however be circumvented using e.g., implicit schemes, cf.~\ref{sec:implicit}.) Thus, in Fig.~\ref{fig:alg_feedback}, the transfer function $G(z)$ constitutes a causally implementable and convergent algorithm provided that $G(z)$ can be factored as
\begin{equation}\label{eq:factorize_G}
  G(z)=  \frac{1}{z-1} \, C(z)
\end{equation}
where $C(z)$ is a proper transfer function with $C(1)\neq 0$.

Consequently, the design of an optimization algorithm amounts to synthesizing a controller $C(z)$ to stabilize the plant that consists of a linear component with transfer function $\frac{1}{z-1}$ composed with the possibly nonlinear component $\Delta_f(\cdot)$.

\subsection{Quadratic optimization as a gain margin problem}\label{sec:quadratic}

We now examine a quadratic objective function~\eqref{eq:quadratic_cost}.
The goal is to design an algorithm that converges fast for all $f\in\mathcal{Q}_{\mu,\ell}$.
In this case, $\Delta_f(\cdot)$ simplifies to $\Delta_f \colon e[t]\mapsto Qe[t]$,
thereby yielding an LTI feedback system in Fig.~\ref{fig:alg_feedback}.

The closed loop system in Fig.~\ref{fig:alg_feedback} has transfer function of the forward path $G(z)Q$, and thereby, the characteristic equation of the closed loop is $\det(I_d+G(z)Q)=0$.
We now let $\{\lambda_1,\lambda_2,\ldots,\lambda_d\}$ be the eigenvalues of $Q$. Thus, the spectral decomposition of $Q$ is $U\diag(\lambda_1,\lambda_2,\ldots,\lambda_d)U\TP$, with $U$ an orthogonal matrix, and the feedback system of
Fig.~\ref{fig:alg_feedback} decomposes into $d$ parallel negative-unity feedback loops
with loop gains $\lambda_i G(z)$, $i\in\{1,2,\ldots,d\}$.

The closed-loop poles determine the convergence rate of the free response. Specifically,
if the zero-input response $\{ \gamma^{-t}e[t]\}_{t=0}^{\infty}$ converges, the closed-loop poles
must have modulus smaller than $\gamma$ and, thereby, the feedback loop with $G(z)$ replaced by $G(\gamma z)$ must be stable. This is the case for each of the $d$ parallel feedback loops.
Thus, since the eigenvalues of $Q$ lie in the interval $[\mu,\ell]$,  $C(\gamma z)$ must stabilize the plant $\frac{\lambda_i}{\gamma z-1}$ for any choice of $\lambda_i\in[\mu,\ell]$.

This is a gain-margin problem which we formalize as follows. We choose the geometric mean of $\mu$ and $\ell$ as the nominal gain of the plant, and we let
\begin{equation}\label{eq:Pgamma}
    P_\gamma (z)\coloneqq \frac{\sqrt{\mu \ell}}{\gamma z-1} \quad\text{and}\quad C_\gamma (z) \coloneqq  C(\gamma z),
\end{equation}
represent transfer functions of the nominal plant and controller, respectively.
For a suitable choice of a controller $C_\gamma$, we seek to extend the range of allowable values for the gain symmetrically with respect to this geometric mean.
The controller needs to guarantee stability for any $P\in\mathcal{P}$, where $\mathcal{P}$ is defined in~\eqref{eq:gain_uncertain} with $P_0=P_\gamma$, $k_1=\sqrt{\kappa^{-1}}$ and $k_2=\sqrt{\kappa}$, for $\kappa=\ell/\mu$.

Thus, $P_\gamma (z)$ is a plant with a single unstable pole at $z=1/\gamma$ and a non-minimum phase zero at $z=\infty$. This aligns with the special case considered in Section~\ref{sec:interpolation}.
Then, by \eqref{eq:unmistable},
\begin{equation*}
    \gamma > \frac{\sqrt{\kappa}-1}{\sqrt{\kappa}+1}=\rmin,
\end{equation*}
being the tightest bound possible, and hence $\rmin$ is the fastest achievable convergence rate for LTI optimization algorithms.
As a final step we construct $G(z)$ that achieves $\rmin$. From \eqref{eq:unmistable_T},
\begin{equation*}
    \bT_\gamma(z) = \frac{\rmin}{\gamma z},
\end{equation*}
and from \eqref{eq:inv_conformal} and \eqref{eq:T0},
\begin{align*}
    G(\gamma z) &= (\mu \ell)^{-1/2} P_\gamma (z) C_\gamma (z) \\
    &=
    \frac{4\rmin \gamma z}{\left(\ell-\mu\right)\left(\rmin^2+\gamma^2 z^2\right) - 2\left(\ell+\mu\right)\rmin\gamma z}.
\end{align*}
In light of the expression for $\rmin$ in \eqref{eq:rho_min},
\begin{equation}\label{eq:heavy_ball_tf}
        G(z) = \frac{4\rmin}{\ell-\mu} \frac{z}{\left(  z -1 \right) \left( z -\rmin^2 \right)}
\end{equation}
is the sought transfer function that achieves the optimal rate $\rmin$. The result shows that  \eqref{eq:heavy_ball_tf} is the transfer function that gives rise to Polyak's heavy-ball method in the $\mathcal{Z}$-domain.





In the sequel we explain Nesterov's stronger result that the same bound holds for the more general time-varying optimization schemes. We will do so by exploring subsequent work by Tannenbaum and his collaborators on gain margin optimization with periodic linear control. We wish to note that Nesterov's conclusion is based on a brilliant construction of an infinite-dimensional quadratic function for which Polyak's asymptotic convergence rate cannot be exceeded (see, e.g., \cite{ugrinovskii2023robust}).
The analysis we provide next shows that Nesterov's bound persists even for finite-dimensional optimization problems.

\section{First-order optimization: periodic schemes}
\label{sec:periodic}

We herein consider schemes where the step size of first-order optimization follows a periodic schedule.
In this, we are motivated by the apparent advantages of periodic control in
extending the allowable gain margin in the context of robust control problems, as we will explain in the subsequent section.


\subsection{Gain margin optimization with periodic control}\label{sec:periodic_control}

Since the mid 1980's, it has been known that the flexibility afforded by linear periodic control has the potential to exceed the performance of LTI control for linear systems, by virtue of relaxing the analytic constraints imposed by non-minimum phase zeros. We next detail the design of periodic control and show that, for the problem at hand, the analytic constraint that arises from causality prevents periodic control from achieving gain margin beyond what can be achieved with LTI control.

The design of periodic control proceeds by ``lifting'' the time-periodic structure to a time-invariant one, albeit for a new set of inputs and outputs that are augmented in dimension, see, e.g., \cite{khargonekar1985robust,francis1988stability}. To explain this, we first define
 the backward shift operator
\begin{equation*}
    U\colon \{x[0],x[1],x[2],\dots\} \mapsto \{0,x[0],x[1],x[2],\dots\},
\end{equation*}
having transfer function $z^{-1}$.
A linear {\em causal} dynamical system $F$ is time-invariant if and only if it commutes with $U$, i.e., $FU=UF$, and it is $n$-periodic if and only if it commutes with $U^n$, i.e., $FU^n=U^nF$. The lifting technique proceeds by vectorizing $n$ successive values of signals, which is realized by introducing the linear operator
\begin{equation*}
    W \colon
    \{x[0],x[1],\dots,x[n],\dots\}
  \mapsto  \left\{  \begin{bmatrix}
        x[0] \\ x[1] \\ \vdots \\ x[n-1]
    \end{bmatrix} ,
    \begin{bmatrix}
        x[n] \\ x[n+1] \\ \vdots \\ x[2n-1]
    \end{bmatrix},\dots \right\} .
\end{equation*}
This process is reversible, hence $W^{-1}$ existing.
It is known that if $F$ is $n$-periodic, then $\widetilde{F}\coloneqq WFW^{-1}$ is time-invariant \cite{khargonekar1985robust}.

Furthermore, if the $n$-periodic system $F$ is (strictly) causal, the $n\times n$ LTI system $\widetilde{F}$ has a transfer function matrix $\widetilde{F}(z)$ that is (strictly) lower triangular matrix at $z=\infty$. To see this, it is helpful to write out the respective action of ${F}$ and $\widetilde{F}$ on time sequences. Hence, if
\begin{equation*}
    F \colon \left\{ u[0],u[1],\dots,u[n],u[n+1],\dots \right\} \mapsto 
    \left\{ y[0],y[1],\dots,y[n],y[n+1],\dots \right\},
    \end{equation*}
    then
    \begin{equation*}
    \widetilde{F} \colon \left\{ \!\begin{bmatrix}
        u[0] \\ u[1] \\ \vdots \\ u[n-1]
    \end{bmatrix},
    \begin{bmatrix}
        u[n] \\ u[n+1] \\ \vdots \\ u[2n-1]
    \end{bmatrix},\dots \right\} \mapsto  
    \left\{ \begin{bmatrix}
        y[0] \\ y[1] \\ \vdots \\ y[n-1]
    \end{bmatrix} ,
    \begin{bmatrix}
        y[n] \\ y[n+1] \\ \vdots \\ y[2n-1]
    \end{bmatrix},\dots \right\}.
\end{equation*}
Because $F$ is causal, $y[i]$ does not depend on $u[j]$ when $i<j$, and hence
the $(i,j)$-th element of the $n\times n$ transfer function matrix $\widetilde{F}(z)$ at $z=\infty$ is zero for $j<i$, i.e., $\widetilde{F}(\infty)$ {\em has a lower triangular structure}. (Likewise, $y[n\tau+i-1]$ does not depend on $u[n\tau+j-1]$, for any $\tau\in \mathbb N$.)  With the same reasoning, an $n\times n$ transfer function matrix $\widetilde{F}(z)$ {\em can be associated with a (strictly) causal linear $n$-periodic single-input single-output system if and only if $\widetilde{F}(\infty)$ is (strictly) lower triangular}.

An LTI system $P$ is trivially $n$-periodic, for any period $n$. In fact, its transfer function
$P(z)=\sum_{k=0}^{\infty} p[t] z^{-t}$ can be written as
\begin{equation*}
    P(z) =
    P_1(z^n) + z^{-1}P_2(z^n) + \dots + z^{-(n-1)}P_n(z^n)
\end{equation*}
and lifted to a multi-input multi-output (MIMO) system with the transfer function \cite{khargonekar1985robust}
\begin{equation}\label{eq:lift_LTI}
        \widetilde{P}(z) \coloneqq WP(z)W^{-1} 
        = \setlength{\arraycolsep}{0pt}
        \begin{bNiceMatrix}[columns-width=auto]
        P_1(z) & z^{-1}P_n(z) & \cdots & z^{-1}P_2(z) \\
        P_2(z) & P_1(z) & \ddots & \vdots \\
        \vdots & \ddots & \ddots & z^{-1}P_n(z) \\
        P_n(z) & \dots & P_2(z) & P_1(z)
    \end{bNiceMatrix}, 
\end{equation}
    where
\begin{equation*}
    P_i(z) = \sum_{t=0}^{\infty} p[tn+i-1] z^{-t}, \quad i=1,2,\dots,n.
\end{equation*}%
Thus, a feedback connection of a SISO LTI plant $P$ with a SISO linear $n$-periodic controller $C$, viewed as an $n$-periodic system, can be simultaneously lifted to form an LTI MIMO feedback system with $n\times n$ respective transfer function matrices.
The stability and input-output induced norms of the feedback system remain the same as for its lifted counterpart.

If $p$ is a pole of an LTI system $P$ then $p^n$ is a pole of the corresponding lifted $\widetilde{P}(z)$. The relation between zeros is substantially different. Finite zeros of $\widetilde{P}(z)$ can be arbitrarily assigned \cite[Section VI]{francis1988stability}, \cite{khargonekar1985robust}, \cite{kabamba1987control} (possibly by introducing a suitable pre-compensation).
As a consequence, periodic control can improve the gain margin beyond what is achievable by LTI control by removing the obstruction due to finite non-minimum phase zeros \cite{khargonekar1985robust,francis1988stability}.  However, zeros at infinity due to strictly causality of $P$ remain, and so does the corresponding obstruction in improving the gain margin.
By extending the theory in Section II-A
to gain margin optimization for MIMO systems, utilizing the approach in Cockburn and Tannenbaum \cite{Cockburn1996multivariable}, we obtain as a corollary the following proposition which is tailored to the case of interest of system without finite non-minimum phase zero.

\begin{proposition}\label{thm:periodic}
    Consider a nominal plant that is strictly causal  without finite non-minimum phase zero. (Due to strict causality, the plant has a non-minimum phase zero at $\infty$.) The maximum gain margin achievable by linear time-periodic controllers is the same as that achievable by LTI controllers.
\end{proposition}

\begin{proof} See Appendix~A.
\end{proof}

\subsection{Optimization algorithms with periodic schedule}

We herein consider optimization algorithms with a periodic schedule. In the simplest form, one can imagine gradient descent with periodically varying step-size, as in \cite{grimmer_provably_2024, altschuler_acceleration_2025}. More generally, optimization algorithms may take advantage of a more sophisticated dynamical periodic structure that extends momentum methods.
In a similar manner as in Section~\ref{sec:opt_control}, periodic optimization can be expressed as tracking with periodic feedback regulation. Thus, in light of the potential of periodic control to improve the gain margin of feedback systems, and thereby the speed of optimization, since the two are intertwined as explained in Section~\ref{sec:quadratic}, also motivated by the results in \cite{grimmer_provably_2024, altschuler_acceleration_2025}, our task is to explore the potential of time periodic algorithms in that respect.

In this section, we first illustrate the application of the lifting technique to express algorithms with a periodic schedule as a corresponding feedback regulation system shown in Fig.~\ref{fig:multi_grad}.
We explain the structure of the ``lifted'' dynamics first for the case of gradient descent with $n$-periodic schedule, and then the structure of a $2$-periodic momentum method.
We show how these can be cast in the configuration shown in Fig.~\ref{fig:multi_grad}.

\begin{figure}[tb]
\centering
	\begin{tikzpicture}[>=latex, line width=0.75pt]
    \tikzstyle{component} = [draw, rectangle, very thick, inner sep=1ex, minimum height=1em]
     \tikzstyle{add} = [draw, circle, inner sep=3pt]

     \node[component] (g) {$\widetilde{G}(z)\otimes I_d$};
     \node[component, above = 3em of g] (d) {$\begin{smallmatrix}
      \Delta_f(\cdot) &  &  \\
       & \ddots &  \\
       &  & \Delta_f(\cdot) \\
     \end{smallmatrix}$};

      \node(a2) [add, left=5em of d] {};

    \draw[->] (a2) --node[above,pos=0.45]{$\begin{bsmallmatrix}
      e_1[\tau] \\ \vdots \\ e_n[\tau]
    \end{bsmallmatrix}$} (d);
    \draw[->] (d) --  ++(10em,0) |- node[above,pos=0.73]{$\begin{bsmallmatrix}
      -\nabla f(x_1[\tau]) \\ \vdots \\ -\nabla f(x_n[\tau])
    \end{bsmallmatrix}$} (g);
    \draw[->] (g) -| node[above,pos=0.27]{$\begin{bsmallmatrix}
      x_1[\tau] \\ \vdots \\ x_n[\tau]
    \end{bsmallmatrix}$} node[pos=0.92,right=-0.3ex]{$-$} (a2);
    \draw[<-] (a2) -- node[pos=0.9,above]{$\begin{bsmallmatrix}
      {x}^\star \\ \vdots \\ {x}^\star
    \end{bsmallmatrix}$} ++(-4em,0);
       \end{tikzpicture}
	\caption{First-order periodic algorithm as a lifted feedback system.}
	\label{fig:multi_grad}
\end{figure}
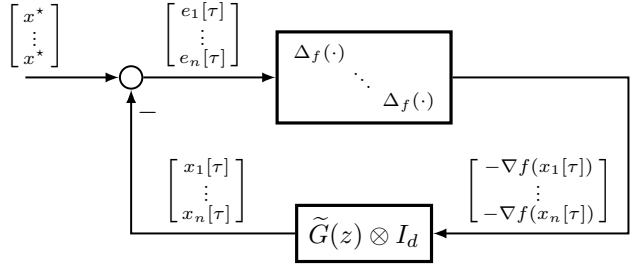

Just as in the LTI algorithms discussed earlier, we assume that at each iteration during the cycle the coordinates $-\nabla f(x[\tau])$ are treated in the same way. Thus, the linear component in the feedback loop has the form of a Kronecker structure $\widetilde{G}(z)\otimes I_d$, in which $\widetilde{G}(z)$ is the $n\times n$ transfer function matrix that characterizes the lifted $n$-periodic algorithm. Note that the time index for the dynamics in Fig.~\ref{fig:multi_grad} is $\tau$, as time indexing contracts accordingly for the lifted system. The nonlinear component in the feedback loop stacks up $n$-copies of $\Delta_f$ defined in \eqref{eq:Delta}.



\subsubsection*{Gradient descent with $n$-periodic scheduling}
	Successive steps of the algorithm make use of step size from a predetermined list of $n$-values, $\{\alpha_1,\ldots,\alpha_n\}$, selected in a periodic manner.
 Thus, the algorithm takes the form
 \begin{equation}\label{eq:periodic}
     x[t+1] = x[t] - \alpha[t] \nabla f(x[t]), \text{ for } t=0,1,2,\ldots
 \end{equation}
 with $\alpha[t]=\alpha_{t\,\mathrm{mod}\,n+1}$.
 In order to transition to a time-invariant structure via the lifting technique, we group $n$ successive values of the state
 \begin{equation*}
     x[t] = x[i-1+n \tau ] \eqqcolon x_i[\tau], \;\text{for}\; i =1,\ldots,n,
 \end{equation*}
 in blocks, that advance in time $\tau$ that represents the integer part of ${t}/{n}$. Thus, we can rewrite the time-periodic relation \eqref{eq:periodic} as the system of time-invariant relations
	\begin{align*}
		x_1[\tau+1] &= x_n [\tau] - \alpha_1 \nabla f(x_n [\tau]) \\
            x_2[\tau+1] &= x_1[\tau+1] - \alpha_2 \nabla f(x_1[\tau+1]) \\
            & \vdots \\
            x_n[\tau+1] &= x_{n-1}[\tau+1] - \alpha_n \nabla f(x_{n-1}[\tau+1]).
	\end{align*}
Referring to Fig.~\ref{fig:multi_grad}, the periodic algorithm can be expressed as a time-invariant feedback loop by taking
	\begin{equation*}
		\widetilde{G}(z) = \frac{1}{z-1} \begin{bmatrix}
		    \alpha_2 & \alpha_3 & \cdots & \alpha_{n} & \alpha_1 \\
                \alpha_2 z & \alpha_3 & \cdots & \alpha_n & \alpha_1 \\
                \vdots & \vdots & \ddots & \vdots & \vdots \\
                \alpha_2 z & \alpha_3 z & \cdots & \alpha_n  & \alpha_1 \\
                \alpha_2 z & \alpha_3 z & \cdots & \alpha_n z & \alpha_1
		\end{bmatrix}.
	\end{equation*}

\subsubsection*{The 2-periodic momentum algorithm}
We further exemplify the process of grouping successive steps where the algorithmic parameters change in a periodic manner by analyzing the momentum method from this point of view, selecting the period to be $2$. Specifically, in Fig.~\ref{fig:multi_grad}, we consider
 \begin{equation*}
     x[t+1] = x[t] - \alpha[t] \nabla f( x[t]) - \eta[t] \nabla f(x[t-1])
     + \beta[t] \left(x[t] - x[t-1] \right)
 \end{equation*}
 with parameters taking values from predetermined lists $\{\alpha_1,\alpha_2\}$, $\{\beta_1,\beta_2\}$, $\{\eta_1,\eta_2\}$ according the periodic scheme
 \begin{equation*}
\alpha[t]=\alpha_{t\,\mathrm{mod}\, 2 +1},\quad
\beta[t]=\beta_{t\,\mathrm{mod}\, 2 +1},\quad
\eta[t]=\eta_{t\,\mathrm{mod}\, 2 +1}.
 \end{equation*}
Once again, setting
 $x_i[\tau]=x[i-1+2 \tau]$, we have
	\begin{align*}
		x_1[\tau+1] &= x_2[\tau]  - \alpha_1 \nabla f( x_2[\tau]) - \eta_1 \nabla f(x_1[\tau])  + \beta_1 \left( x_2[\tau] - x_1[\tau] \right) \\
  x_2[\tau+1] &= x_1[\tau+1]  - \alpha_2 \nabla f( x_1[\tau+1]) - \eta_2 \nabla f(x_2[\tau])  + \beta_2 \left( x_1[\tau+1] - x_2[\tau] \right).
	\end{align*}
Applying the $\mathcal Z$-transform on the above, we readily see that
 \begin{equation*}
     \widetilde{G}(z) = \frac{1}{(z-1) (z-\beta_1 \beta_2 )}
  \begin{bmatrix}
      z+\beta_2 & 1+\beta_1 \\
      z (1+\beta_2) & z+\beta_1
  \end{bmatrix}
  \begin{bmatrix}
      \eta_1 &\!\! \alpha_1 \\ \alpha_2 z &\!\! \eta_2
  \end{bmatrix}.
 \end{equation*}






\subsection{Limitation of algorithms with periodic schedules}\label{sec:periodic_limitation}

In spite of potential advantages of periodic control in improving robustness of control systems,  periodic scheduling fails to improve convergence in optimization algorithms.
%
To see this, as before, the control problem in Fig.~\ref{fig:multi_grad} can be viewed as a tracking problem for a  step reference that is now vector-valued with magnitude
\begin{equation}\label{eq:xstar_repeat}
	\mathbf{x}^\star \coloneqq \allone{n} \otimes x^\star.
\end{equation}
Once again, $x^\star$ represents the solution to \eqref{eq:opt_problem}.
Similar to the analysis of LTI algorithms, $\widetilde{G}(z)$ must have an accumulator.

However, at present, $\widetilde{G}(z)$ is a MIMO transfer function and, thereby, has poles and zeros that are manifested possibly along specific directions. The directionality of poles and zeros can be captured using coprime fractional representations of $\widetilde{G}(z)$ (see, e.g., \cite{vidyasagar2011control}, \cite{chen_logarithmic_2000}). To this end, consider the ``left fraction''
\begin{equation}\label{eq:coprime_factorization}
    \widetilde{G}(z)={\widetilde{M}(z)\!}^{-1} \widetilde{N}(z),
\end{equation}
with ``denominator'' and ``numerator'' matrices $\widetilde{M}(z), \widetilde{N}(z)$ of suitable dimensions and entries in $\mathcal{RH}_\infty$, which are said to be ``coprime'' by virtue of the fact that they satisfy the equation
$\widetilde{N}X+\widetilde{M}Y=I$ for suitable matrices $X$, $Y$ with entries in $\mathcal{RH}_\infty$. The matrix $\widetilde{G}(z)$ that we employ will be square and of full normal rank, in which case,
any pole $p$ of $\widetilde{G}(z)$ satisfies $\widetilde{M}(p)v=\mathbf{0}$ for some vector $v$ \cite[Section II-C]{chen_logarithmic_2000}; this vector $v$ represents a direction that excites $p$-related dynamics.


The following proposition is the counterpart of Proposition~\ref{prop:LTI} for periodic algorithms. It explains that the accumulator, that corresponds to a pole at $z=1$, must act along a direction $\allone{n}$ that is dictated by the nature of the reference signal in \eqref{eq:xstar_repeat}.

\begin{proposition}\label{prop:Minv_1}
Let $f$ be differentiable and let the minimizer $x^\star\in\IR^d$ be a unique stationary point of $f$.
     Then, the $n$-periodic optimization algorithm characterized by the $n\times n$ transfer function matrix $\widetilde{G}(z)$ converges to the (unknown) minimizer  $x^\star$ if and only if the feedback system in Fig.~\ref{fig:multi_grad} is asymptotically stable for exogenous input set to zero, and $\widetilde{M}(z)$ in the coprime factorization \eqref{eq:coprime_factorization} satisfies the interpolation condition
	\begin{equation}\label{eq:Minv_1}
 {\widetilde{M}(1)}\allone{n}=\mathbf{0}_{n}.
	\end{equation}
\end{proposition}

\begin{proof}
The feedback equations for Fig.~\ref{fig:multi_grad} in the $\mathcal{Z}$-domain read
\begin{equation*}
     \widehat{\mathbf{e}} (z) = \tfrac{z}{z-1} \mathbf{x}^\star -
    \bigl(\widetilde{G}(z)\otimes I_d\bigr)
    \widehat{\Delta_f(\mathbf{e})} (z),
\end{equation*}%
where $\widehat{\mathbf{e}} (z)$ and $\widehat{\Delta_f(\mathbf{e})} (z)$ respectively denote the $\mathcal{Z}$-transform of the input and output of the nonlinear component, i.e.,
\begin{equation*}
    \widehat{\Delta_f(\mathbf{e})} (z) \coloneqq \begin{bmatrix}
        \widehat{\Delta_f(e_1)}(z) \\ \vdots \\ \widehat{\Delta_f(e_n)}(z)
    \end{bmatrix},\quad
    \widehat{\mathbf{e}} (z) \coloneqq
    \begin{bmatrix}
        \widehat{e}_1(z) \\ \vdots \\ \widehat{e}_n(z)
    \end{bmatrix}.
\end{equation*}
	Pre-multiplying both sides by $\left( z - 1 \right)\bigl({\widetilde{M}(z)}\otimes I_d\bigr)$, we get
 \begin{equation*}
              (z-1)
    \Bigl( \bigl(\widetilde{N}(z)\otimes I_d\bigr)
    \widehat{\Delta_f(\mathbf{e})} (z) +
    \bigl({\widetilde{M}(z)}\otimes I_d\bigr)
   \widehat{\mathbf{e}} (z) \Bigr) = {z} \bigl({\widetilde{M}(z)}\otimes I_d\bigr) \mathbf{x}^\star.
     \end{equation*}
 Since $\widetilde{M}(z), \widetilde{N}(z)\in\mathcal{RH}_\infty$,
 if ${e}_i[\tau]$ converges to zero for all $i=1,\dots,n$, then the left-hand side of the above equation tends to $\mathbf{0}_{nd}$ as $z\to 1$. Hence,
 \begin{equation*}
	\mathbf{0}_{nd} = {\bigl( \widetilde{M}(1) \otimes I_d \bigr)} \left(\allone{n} \otimes x^\star\right) = \left( {M(1)}\allone{n}  \right)\otimes x^\star
 \end{equation*}
 holds for all ${x}^\star\in\IR^d$, which establishes \eqref{eq:Minv_1}. The converse is a consequence of the internal model principle \cite{francis1975internal, francis_internal_1976}.
\end{proof}

Similar to the analysis of LTI algorithms, any periodic algorithm must avoid circular dependence as well, and accordingly, $\widetilde{G}(z)$ must assume a structure that enforces causality in carrying out needed computations. This is stated next.

\begin{proposition}\label{prop:strictly_G}
	An $n$-periodic algorithm characterized by the $n\times n$ transfer function matrix $\widetilde{G}(z)$ avoids circular dependence if and only if
$\widetilde{G}(\infty)$ is a strictly lower triangular matrix.
\end{proposition}
\begin{proof}
The proof is similar to the reasoning in Section~\ref{sec:periodic_control} that explains the lower triangular structure of $\widetilde{F}(\infty)$ for a causal periodic system $F$.
From the construction of the lifted system, $x_i[\tau+1]$ is computed before $x_j[\tau+1]$ for $1\leq i < j\leq n$. Thus, to avoid circular dependence, the computation of $x_i[\tau+1]$ cannot exploit $x_j[\tau+1]$ for $i \leq j$. The proof is complete by noting that the
$\left( i,j\right)$-th element of $\widetilde{G}(\infty)$ is nonzero if and only if computing $x_i[\tau+1]$ requires $\nabla f(x_j[\tau+1])$.
\end{proof}


For the case of a quadratic cost function \eqref{eq:quadratic_cost}, as in Section~\ref{sec:quadratic}, the design of an optimization algorithm for all $f\in\mathcal{Q}_{\mu,\ell}$ amounts to synthesizing a controller $C_\gamma$ that achieves gain margin $\kappa=\ell/\mu$ for the nominal plant $P_\gamma$ in \eqref{eq:Pgamma}. As for designing $n$-periodic algorithms, the controller $C_\gamma$ is $n$-periodic. The flexibility afforded by periodic control has the promise to exceed the gain margin of time-invariant control. However, Proposition~\ref{thm:periodic} states that if the nominal plant is strictly causal without no finite non-minimum phase zero, the maximum gain margin achievable by linear time-periodic controllers is the same as that achievable by LTI controllers. The nominal plant $P_\gamma$ shares the same structure.

In various existing works, the optimality of the heavy-ball method is established from different perspectives.
In \cite[Theorem 2.1.13]{nesterov2018lectures}, a ``worst-case'' quadratic function is constructed to show that any algorithm that generates a sequence of points $\{x[t] \}$ such that
\begin{equation*}
    x[t+1] - x[0] \in  \mathrm{span} \{ \nabla f(x[0]),\dots, \nabla f(x[t])\},
\end{equation*}
which includes any linear time-varying algorithm, cannot exceed the convergence rate $\rmin$ defined in \eqref{eq:rho_min}. The fastest convergence rate \eqref{eq:rho_min} is achieved by Polyak's heavy-ball method, which is a time-invariant algorithm. However, the dimension of the ``worst-case'' quadratic function grows with the number of iterations to be considered. In \cite{arjevani_lower_2016}, a dimension-independent lower bound of convergence rate is established for any algorithm that only uses current and one-step past information, which includes the heavy-ball method. The lower bound is shown to be $\rmin$.

Our analysis explains why time-varying algorithms fail to improve the convergence rate---the fundamental obstacle being the zero at infinity due to causality. Thus, a faster rate than $\rmin$ can only be achieved by a nonzero direct feedthrough gain $G(\infty)$, as in the algorithmic structure discussed in Section~\ref{sec:implicit_section} on implicit algorithms.

\subsection{Online optimization}

As a final note on periodic schemes, we discuss online optimization.
To this end, we consider the time-varying function
\begin{equation*}
    f_t(x) = \frac{1}{2} x\TP Q x - q_t\TP x,
\end{equation*}
where $Q=Q\TP$ with $\spec(Q)\subset[\mu,\ell]$, as before, and $q_t\in\IR^d$ varies with time. We consider the practically important case where ${\{q_{t} \}}_{t=0}^{\infty}$ has a rational $\mathcal{Z}$-transform $\widehat{q}(z)$ with poles $\pi_{1},\dots,\pi_{r}$ on $\partial\mathbb{D}$, and seek algorithms that asymptotically track the instantaneous minimizer $x_t^{\star}=Q^{-1} q_t$, see \cite{bastianello_internal_2024, wu2025fundamental, van2025temporal}.
As in Section~\ref{sec:opt_control}, the algorithm-design problem can be cast as a robust tracking problem with reference signal $x^{\star}[t]=x_t^{\star}$, which inherits the same modal content as $q_t$.

By the internal model principle, asymptotic tracking of $x_t^\star$ requires that $G(z)$ has poles at $\pi_{1},\dots,\pi_{r}$. The nominal plant $P_{\gamma}(z)$ in \eqref{eq:Pgamma} is accordingly augmented to
\begin{equation*}
    P_{\gamma}(z) = \frac{\sqrt{\mu\ell}}{z} \prod_{i=1}^{r} \frac{z}{\gamma z-\pi_i},
\end{equation*}
which has unstable poles at $z=\pi_i/\gamma$ for $i=1,\dots,r$ and matches the special case treated in Section~\ref{sec:interpolation}.
Then, by \eqref{eq:unmistable},
\begin{equation*}
    \gamma^{r} = \prod_{i=1}^{r} {\Bigl|\frac{\gamma}{\pi_i}\Bigr|} > \frac{\sqrt{\kappa}-1}{\sqrt{\kappa}+1},
\end{equation*}
and therefore the fastest achievable convergence rate is
\begin{equation}\label{eq:rhoT}
    \rho_{\mathrm{T}} \coloneq {\left( \frac{\sqrt{\kappa}-1}{\sqrt{\kappa}+1} \right)}^{{1}/{r}}.
\end{equation}
Mutatis mutandis, as in Section~\ref{sec:quadratic}, we obtain $G(z)$ that attains this fastest rate $\rho_{\mathrm{T}}$:
from \eqref{eq:unmistable_T},
\begin{equation*}
    \bT_\gamma(z) = \frac{\rho_{\mathrm{T}}^{r}}{\gamma^{r}} \frac{\prod_{i=1}^{r} ( \gamma z -  \pi_{i} ) - \gamma^{r} \prod_{i=1}^{r} ( z -  \gamma \pi_{i} )}{\gamma^{r} \prod_{i=1}^{r} (\gamma z - \pi_{i} ) - \prod_{i=1}^{r} ( z -  \gamma \pi_{i} )}.
\end{equation*}
In light of \eqref{eq:inv_conformal} and \eqref{eq:T0}, with $\gamma = \rho_{\mathrm{T}}$, a lengthy calculation gives
\begin{equation*}
        G(z) = \frac{\rho_{\mathrm{T}}^{2r} \prod_{i=1}^r{( z -\pi_{i}  )}^{2} - (1+\rho_{\mathrm{T}}^{2r}) \prod_{i=1}^{r}{(z -\pi_{i}) (z -\rho_{\mathrm{T}}^2 \pi_{i}) } + \prod_{i=1}^r{( z -\rho_{\mathrm{T}}^2 \pi_{i}  )}^{2}}{4\ell\mu {\bigl(\sqrt{\ell}+\sqrt{\mu}\bigr)\!}^{-2} \prod_{i=1}^{r}{(z -\pi_{i}) (z -\rho_{\mathrm{T}}^2 \pi_{i})} }.
\end{equation*}
Note that when $r=1$ and $\pi_{1}=1$, $G(z)$ reduces to \eqref{eq:heavy_ball_tf}.

In \cite{wu2025fundamental}, a similar approach was used to analyze the convergence rate for the online optimization problem when $x_{t}^\star$ is polynomial in $t$, i.e., $\pi_{i}=1$, $i=1,\dots,r$, and obtained the same expression \eqref{eq:rhoT}.
Specializing $G(z)$ given above to $\pi_{i}=1$ gives the expression proposed in \cite{wu2025fundamental} as well.
We finally note that Proposition~\ref{thm:periodic} implies that time-periodic algorithms cannot exceed the rate $\rho_{\mathrm{T}}$ in this online optimization problem.

\section{First-order optimization: implicit algorithms}\label{sec:implicit_section}

As highlighted in the previous section, the strict causality of the nominal $P$ limits the achievable convergence rate of optimization algorithms.
Thus, we now explore the possibility of 
bypassing this obstacle
by allowing the optimization to utilize information in a causal, but not strictly causal, manner.
%
%
This amounts to allowing the corresponding transfer function $G (z)$, from $u [t] \coloneqq -\nabla f (x [t])$ to $x [t]$, to have a  non-vanishing direct feedthrough gain $\delta = G (\infty)$, thereby 
requiring a suitable adaptation of Tannenbaum's gain margin analysis.
We refer to the corresponding algorithms as \emph{implicit},
to highlight the usage in numerical analysis literature of the term when both $x [t+1]$ and $\nabla f (x [t+1])$ appear simultaneously in iterative steps.

\subsection{Implicit algorithms for quadratic functions}\label{sec:bi_quadratic}

As in Section~\ref{sec:quadratic}, the design of an implicit algorithm requires synthesizing controller $C_\gamma$ to achieve gain margin $\kappa=\ell/\mu$ for the nominal plant $z P_\gamma(z)$ which is proper, but no longer strictly proper.
Since $z P_\gamma(z)$ is minimum phase, the achievable gain margin can now be made arbitrarily large. Thereby, for quadratic functions, implicit algorithms can achieve convergence rates that depend on the direct feedthrough gain $\delta=G(\infty)$, and can in principle be made arbitrarily fast.

\begin{theorem}\label{thm:bi_quadratic}
    Consider the optimization algorithm characterized by the transfer function $G(z)$ with the direct feedthrough gain $\delta=G(\infty)$. The algorithm achieves $\rho$-convergence for all quadratic functions $f\in\mathcal{Q}_{\mu,\ell}$ if and only if
\begin{equation}\label{eq:bi_quadratic}
        \rho \geq
        \frac{\sqrt{\kappa + \ell\delta } - \sqrt{1+\ell\delta}}{\sqrt{\kappa+\ell\delta } + \sqrt{1+\ell\delta}}.
    \end{equation}
\end{theorem}
\vspace{1ex}
\begin{proof}
    The proof is similar to the derivation in Section~\ref{sec:quadratic} that constructs the heavy-ball method using the gain margin maximizer. The key difference lies in the interpolation condition on the nominal complementary sensitivity $T_\gamma(z)$. Since $G(\infty)=\delta$, we have
    \begin{equation*}
        T_\gamma(\infty)=\frac{\delta\sqrt{\mu\ell}}{1+\delta\sqrt{\mu\ell}}
    \end{equation*}
    and consequently, via \eqref{eq:bfT},
    \begin{equation*}
        \mathbf{T}_\gamma (\infty) = \frac{\sqrt{1+k_2 \delta \sqrt{\mu\ell}} - \sqrt{1+ k_1 \delta \sqrt{\mu\ell}}}{\sqrt{1+ k_2 \delta \sqrt{\mu\ell}} + \sqrt{1+k_1 \delta \sqrt{\mu\ell}}}.
    \end{equation*}
    Other requirements on $\mathbf{T}_\gamma(z)$ remain the same as those for strictly causal algorithms, i.e., $\mathbf{T}_\gamma (z)$ must be analytic in $\mathbb{D}^c$ with range in $\mathbb{D}$ and satisfy
    $\mathbf{T}_\gamma (\gamma^{-1}) = g(k_1,k_2)$. Using the Nevanlinna--Pick theory reviewed in Section~\ref{sec:interpolation}, such a $\mathbf{T}_\gamma(z)$ exists if and only if
    \begin{equation*}
        \gamma >
        \frac{\sqrt{k_2+ k_1 k_2 \delta \sqrt{\mu\ell}} - \sqrt{k_1+k_1 k_2 \delta \sqrt{\mu\ell}}}{\sqrt{k_2+k_1 k_2 \delta \sqrt{\mu\ell}} + \sqrt{k_1+ k_1 k_2 \delta \sqrt{\mu\ell}}}.
    \end{equation*}
    With the right-hand side of the above inequality being the tightest bound possible, the proof is completed by substituting $k_1=\sqrt{\kappa^{-1}}$ and $k_2=\sqrt{\kappa}$, for $\kappa=\ell/\mu$.
\end{proof}

Note that the statement of Theorem~\ref{thm:bi_quadratic} allows for $G(z)$ to be strictly causal ($\delta=0$).
In this case,  Theorem~\ref{thm:bi_quadratic} recovers the fastest achievable convergence rate $\rmin$ defined in \eqref{eq:rho_min}, i.e., the convergence rate of the heavy-ball method. As $\delta$ increases, the convergence rate can be made arbitrarily fast. Negative values for $\delta$ degrade performance and even lead to instability.
For a specified convergence rate $\rho\in(0,1)$, Theorem~\ref{thm:bi_quadratic} provides the minimum direct feedthrough gain needed,
\begin{equation}\label{eq:delta_rho}
     \delta_{\rho} \coloneqq \frac{{\left(1-\rho\right)}^2 \kappa - {\left(1+\rho\right)}^2 }{4\rho\ell}.
\end{equation}
Note that $\delta_{\rho}$ is a decreasing function of $\rho\in(0,1)$ and that $\delta_{\rho}=0$ at $\rho=\rmin$. Thus, to achieve a convergence rate $\rho$ faster than $\rmin$, a strictly positive direct feedthrough gain $\delta$ is required, and that a faster convergence rate $\rho$ demands a larger direct feedthrough gain $\delta$. In other words, to better the heavy-ball rate $\rmin$, not only is it sufficient, but also necessary for $G(z)$ to be causal but not strictly causal.

When $\delta=\delta_{\rho}$, the transfer function of the algorithm that achieves the $\rho$-convergence can be synthesized using the Nevanlinna--Pick algorithm.
The resulting transfer function $G(z)$ is unique and is given by
\begin{equation}\label{eq:bi_heavy_ball_tf}
    G(z) = \frac{\delta_{\rho} z^2 + \beta z + \delta_{\rho}\rho^2}{\left(z-1\right)\left(z-\rho^2\right)},
\end{equation}
where
\begin{equation*}
    \beta = \frac{4+2\delta_{\rho}\left(\ell+\mu\right)}{{\left( \sqrt{\ell+\mu\ell\delta_\rho } + \sqrt{\mu+\mu\ell\delta_\rho } \right)\hspace{-0.1em}}^2}.
\end{equation*}
The corresponding time-domain iteration is
\begin{equation}\label{eq:implicit_QP}
     x[t+1] = x[t] + \rho^2 \left(x[t]-x[t-1]\right) - \delta_{\rho} \nabla f(x[t+1]) 
     - \beta \nabla f(x[t]) - \delta_{\rho} \rho^2 \nabla f(x[t-1]).
\end{equation}
Substituting the expression $\nabla f(x) = Qx-q$ for the gradient into \eqref{eq:implicit_QP} provides a causally
implementable form for the numerical scheme, namely, the
\emph{implicit heavy-ball algorithm}
\begin{equation}\label{eq:bi_heavy_ball}
    x[t+1] = x[t] + \rho^2 \left(x[t]-x[t-1]\right) 
     - (\delta_{\rho}+\delta_{\rho} \rho^2+\beta) {\left(I_d + \delta_{\rho}Q\right)}^{-1} \left(Qx[t]-q\right).
\end{equation}


In fact, the algorithm \eqref{eq:bi_heavy_ball} can be seen to interpolate the heavy-ball method and Newton's method with the choice of the parameters $\delta_\rho$ or $\rho$, that are coupled via \eqref{eq:delta_rho}. Clearly, when $\rho=\rmin$, the direct feedthrough gain $\delta_{\rho}=0$, \eqref{eq:bi_heavy_ball_tf} reduces to the transfer function of the heavy-ball method \eqref{eq:heavy_ball_tf}, and the algorithm \eqref{eq:bi_heavy_ball} reduces to the corresponding method \eqref{eq:heavy_ball}. On the other hand, it is well-known that Newton's method only needs one iteration to obtain the minimizer of a quadratic function, which is reflected in the fact that the algorithm \eqref{eq:bi_heavy_ball} approaches Newton's method as $\rho\to 0$. To see this, we rewrite the iteration \eqref{eq:bi_heavy_ball} as
\begin{equation}\label{eq:bi_heavy_ball_inv}
    x[t+1] = x[t] + \rho^2 \left(x[t]-x[t-1]\right) - (1 + \rho^2+\beta \delta_{\rho}^{-1}) {\left( \delta_{\rho}^{-1} I_d +Q\right)}^{-1} \left(Qx[t]-q\right),
\end{equation}
and therefore, as $\rho\to 0$, we have that $\delta_{\rho}^{-1}\to 0$, $\beta \delta_{\rho}^{-1}\to 0$. Consequently, the above iteration turns into Newton's method. As such, the heavy-ball algorithm and Newton's method can be regarded as limiting cases of the implicit algorithm \eqref{eq:bi_heavy_ball}.

\subsection{Applications to ill-conditioned quadratic functions}\label{sec:bi_ill_quadratic}

When the condition number $\kappa$ of $Q$ exceeds a value $\kappa_{\mathrm{m}}$ that allows acceptable accuracy in solving $Qx=q$, Newton's method fails while the heavy-ball method converges slowly. In such cases, \eqref{eq:bi_heavy_ball} proves a viable alternative. It draws benefits from the condition number of $I_d + \delta_{\rho}Q$, which is significantly improved, for a suitably small $\delta_{\rho}$, and at the same time, the rate of convergence can be made substantially faster than that of the heavy-ball method.
Thus while \eqref{eq:bi_heavy_ball} or \eqref{eq:bi_heavy_ball_inv} can be viewed as an \emph{iterative refinement with momentum} to solve $Qx=q$,
it goes beyond other types of iterative refinement that mostly focus on addressing the ill-condition of $Q$, e.g., \cite{wilkinson_iterative_1971,beik_iterative_2018}.

Specifically, to achieve a convergence rate $\rho$, we  select $\delta_{\rho}$ so that
\begin{equation*}
    \frac{1+\delta_{\rho} \ell}{1+\delta_{\rho} \mu} \leq \kappa_{\mathrm{m}} \leq \kappa = \frac{\ell}{\mu}.
\end{equation*}
Direct computation shows that $\delta_\rho\geq 0$ needs to satisfy
\begin{equation}\label{eq:delta_rho_bound}
    0\leq \delta_{\rho} \leq \frac{\kappa_{\mathrm{m}}-1}{\ell - \mu \kappa_{\mathrm{m}}} = \frac{1}{\mu} \frac{\kappa_{\mathrm{m}}-1}{\kappa - \kappa_{\mathrm{m}}}.
\end{equation}
Then, from \eqref{eq:delta_rho}, the achievable convergence rate $\rho$ satisfies
\begin{equation*}
    \rho \geq \rho_{\mathrm{m}} \coloneqq \frac{\sqrt{\kappa / \kappa_{\mathrm{m}}} - 1}{\sqrt{\kappa / \kappa_{\mathrm{m}}} + 1}.
\end{equation*}
In fact, the bound $\rho_{\mathrm{m}}$ is the fastest rate that can be achieved by implicit algorithms for ill-conditioned quadratic functions with $\kappa\geq\kappa_{\mathrm{m}}$, and is attained when $\delta_{\rho}$ equals the upper bound in \eqref{eq:delta_rho_bound}.

We finally note that both the heavy-ball method and Newton's method can be directly contrasted with the above.  The heavy-ball method can be seen as the special case where $\kappa_{\mathrm{m}}=1$, as it avoids computing matrix inverse.
In this case, the upper bound in \eqref{eq:delta_rho_bound} is $0$ and $\rho_{\mathrm{m}}=\rmin$.
On the other hand, Newton's method can be seen as the special case where $\kappa_{\mathrm{m}}=\kappa$, since it requires the inverse of $Q$. In this case, the upper bound in \eqref{eq:delta_rho_bound} is $+\infty$ and $\rho_{\mathrm{m}}=0$.

In summary, redefine
\begin{equation*}
    \rho_{\mathrm{m}} \coloneqq
    \max\left\{ 0, \frac{\sqrt{\kappa / \kappa_{\mathrm{m}}} - 1}{\sqrt{\kappa / \kappa_{\mathrm{m}}} + 1} \right\}.
\end{equation*}
This is the fastest convergence rate that can be achieved by an algorithm
to minimize a quadratic function $f$ with condition number $\kappa$, to be
implemented with a solver that can reliably compute inverses or Cholesky factorization of matrices with condition number less than $\kappa_{\mathrm{m}}$.

\section{Non-quadratic functions}\label{sec:nonquadratic}

For general cost functions, the $\Delta_f(\cdot)$ in Fig.~\ref{fig:alg_feedback} is a memoryless nonlinear element, but Proposition~\ref{prop:LTI} still holds. Moreover, if $\Delta_f(\cdot)$ is sector-bounded as defined below, the feedback scheme in Fig.~\ref{fig:alg_feedback} is a Lur'e system and its stability can be ascertained using the well-known \emph{circle criterion} \cite{khalil2002nonlinear}. We show in this section that for such cost functions, the non-strict causality can still be exploited to accelerate convergence rate, by making use of implicit algorithms.
We begin by defining two important types of nonlinear structure, see also \cite{lessard2016analysis}.

\begin{definition}
    A map $\phi\colon \IR^d\to\IR^d$ is \emph{sector-bounded} in $\left[ k_1,k_2 \right]$, where $ k_2 > k_1 \geq 0$, if
 \begin{equation}\label{eq:sector_def}
 \left< \phi(x) - k_1 x , \phi(x) - k_2 x \right> \leq 0 \quad \forall x \in\IR^d,
 \end{equation}
and is {\em slope-restricted} in $\left[ k_1,k_2 \right]$, where $k_2>k_1\geq 0$, if
  \begin{displaymath}
  \left<  \delta y - k_1 \delta x , \delta y - k_2 \delta x \right>
     \leq 0 \quad \forall x_1,x_2 \in\IR^d,
 \end{displaymath}
 where $\delta y \coloneqq \phi(x_1)-\phi(x_2)$ and $\delta x \coloneqq x_1-x_2$.
\end{definition}

These two types of nonlinearity are closely related as explained in the following two lemmas that follow easily from the definitions (cf.\ \cite{lessard2016analysis}).

\begin{lemma}\label{lem:slope_sector}
    A map $\phi\colon \IR^d\to\IR^d$ is slope-restricted in $\left[ k_1,k_2\right]$, if and only if $D_y\phi(x)\coloneqq \phi(y) - \phi(y -x)$ is sector-bounded in $\left[ k_1,k_2 \right]$ for all $y\in \IR^d$.
\end{lemma}

Bounding the slope of nonlinear elements has been key in stability analysis of nonlinear control systems \cite{zames_stability_1968}, and this is closely related to conditions imposed in optimization theory and stated in Definition~\ref{def:convex}.

\begin{lemma}\label{lem:slope_convex}
    A differentiable function $f\colon \IR^d\to\IR$ is both $\mu$-strongly convex and $\ell$-Lipschitz smooth (i.e., $f\in\mathcal{F}_{\mu,\ell}$) if and only if its gradient $\nabla f$ is slope-restricted in $[\mu,\ell]$.
\end{lemma}

\begin{definition}
    The set of functions whose gradient differences $\Delta_f(\cdot)$ defined in \eqref{eq:Delta} are sector-bounded in $[\mu,\ell]$ is denoted by $\mathcal{S}_{\mu,\ell}$.
\end{definition}

From Lemmas~\ref{lem:slope_sector} and~\ref{lem:slope_convex}, the set $\mathcal{S}_{\mu,\ell}$ is a proper superset of $\mathcal{F}_{\mu,\ell}$, which may contain non-convex functions that may or may not be Lipschitz smooth, a fact observed in \cite{fazlyab_analysis_2018, Ugrinovskii2022} in analyzing optimization algorithms for nonconvex functions. Consider, for example, the univariate function $f$ given by
\begin{equation}\label{eq:sector_1d}
    f(x) = 0.5 \alpha x^2 - 0.5 \beta g(x), \quad
\end{equation}
where $\alpha>\beta>0$, and $g$ is defined as
\begin{equation*}
    g(x) = \begin{cases}
        {\left| x-1 \right|}^{{1+\nu}} + (1+\nu) x - 1 & \text{if } x \geq 0, \\
        0.5 x^{2} +  0.5 \sin(x^{2}) & \text{if } x< 0,
    \end{cases}
\end{equation*}
with $\nu\in(0,1)$.
It is a simple exercise to show that
\begin{equation*}
    \Delta_f(e) = - f^{\prime} (x^\star-e) = \alpha e + 0.5 \beta g^{\prime}(-e)
\end{equation*}
is sector-bounded in $\left[\alpha-\beta,\alpha\right]$. As such,
$f\in \mathcal{S}_{\mu,\ell}$. On the other hand, it is also easy to see from the graphs of $f$ and $\Delta_f$, as depicted in Fig.~\ref{fig:sector_1d}, that $f$ is neither
star-convex (a relaxation of convexity introduced by Nesterov and Polyak in \cite{nesterov_cubic_2006}) nor Lipschitz smooth.

\begin{figure}[tb]
    \centering
        \begin{subfigure}[T]{0.493\linewidth}
     \centering
            \begin{tikzpicture}[]
    \begin{axis}[scale=0.6, axis lines=middle, xlabel=$x$, ylabel=$f(x)$, xmin=-5, xmax=5, ymax=70, ymin=-5,
    ticks=none,
x label style={anchor=north, xshift=-0.7ex},
     y label style={anchor=east, yshift=-1ex, xshift=0.3ex}]


    \addplot[domain=-5:0, smooth, thick, samples=200]{8/2*x^2 - 6/2*(0.5*x^2 + 0.5*sin(deg(x^2)))};
    \addplot[domain=0:5, smooth, thick, samples=200]{8/2*x^2 - 6/2*(abs(x-1)^1.2 + 1.2*x - 1)};

    \addplot[domain=-5:5, smooth, thick, dashed]{4*x^2};

    \addplot[domain=-5:5, smooth, thick, dashed]{x^2};

    \node[inner sep=0.1ex, anchor=south east, yshift=-1.3ex] at(axis cs:5,6) (m) {$\frac{\mu}{2} x^2$};

    \node[inner sep=0.1ex, anchor=south east, xshift=0.1ex] at(axis cs:3,35) (l) {$\frac{\ell}{2} x^2$};

    \end{axis}
\end{tikzpicture}
        \end{subfigure}
     \hfill
         \begin{subfigure}[T]{0.493\linewidth}
     \centering
            \begin{tikzpicture}[]
 \begin{axis}[scale=0.6, axis lines=middle, xlabel=$e$, ylabel=$\Delta_f(e)$, xmin=-5, xmax=5, ymax=40, ymin=-40,
    ticks=none,
x label style={anchor=north, xshift=-0.7ex},
     y label style={anchor=east, yshift=-1ex, xshift=0.3ex},]

    \addplot[domain=-5:0, smooth, thick, samples=200]{8*x-6/2*(x+x*cos(deg(x^2)))};
    \addplot[domain=0:5, smooth, thick, samples=200]{8*x-6/2*1.2*(abs(x-1)^0.2 * sign(x-1) + 1)};

    \addplot[domain=-5:5, thick,dashed]{2*x};

    \addplot[domain=-5:5, thick,dashed]{8*x};

    \node[inner sep=0.1ex] at(axis cs:1.5,-20) (m) {$\mu e$};
    \draw[->, thick] (m) to[out=0,in=-80] (axis cs:4,8);

    \node[inner sep=0.1ex] at(axis cs:1,30) (l) {$\ell e$};
    \draw[->, thick] (l) to[out=0,in=120] (axis cs:3.4,27.2);
    \end{axis}
\end{tikzpicture}
        \end{subfigure}
        \caption{The univariate function $f(x)$ given in \eqref{eq:sector_1d} with $\alpha=8$, $\beta=6$, and $\nu=0.2$, where $\mu=\alpha-\beta$ and $\ell=\alpha$.}
        \label{fig:sector_1d}
    \end{figure}
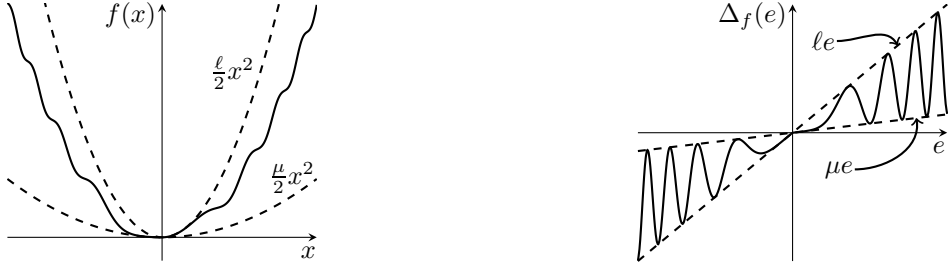

\subsection{Implicit algorithms}\label{sec:implicit}

We now consider the optimization problem \eqref{eq:opt_problem} with $f\in\mathcal{S}_{\mu,\ell}$. The feedback system in Fig.~\ref{fig:alg_feedback} is a Lur'e system where nonlinear component $\Delta_f(\cdot)$ is sector-bounded in $[\mu,\ell]$. From Proposition~\ref{prop:LTI}, the optimization algorithm characterized by $G(z)$ achieves $\rho$-convergence if and only if the feedback system in Fig.~\ref{fig:alg_feedback} is $\rho$-stable and $G(z)$ has an accumulator. As a Lur'e system, the system's stability can be ascertained using the circle criterion. Before stating the circle criterion, we introduce the concept of strict positive realness.

\begin{definition}
    A square rational matrix function $\Psi(z)$ is said to be \emph{strictly positive real} if $\Psi(z)$ is analytic in $\bD^c\setminus\partial\bD$ with strictly positive definite Hermitian part, i.e.,  for some $\epsilon>0$,
    \begin{equation*}
         \Psi(z) + \Psi\HP(z) \succeq \epsilon I, \quad \forall z\in \bD^c\setminus\partial\bD.
    \end{equation*}
\end{definition}

\vspace*{1ex}

The right half of the complex plane can be conformally mapped onto the unit disk $\mathbb{D}$ by the transformation
\begin{equation}\label{eq:conformal}
    \psi \colon v \mapsto u=\frac{1-v}{1+v}.
\end{equation}
Thus, a matrix $\Psi$ has positive definite Hermitian part if and only if the matrix $\Phi = {\left(I+\Psi\right)}^{-1}\left(I-\Psi\right)$ is contractive, i.e., $\overbar{\sigma}(\Phi)<1$. And since the inverse map of $\psi$ is itself, the matrix $\Psi$ is contractive if and only if the matrix $\Phi$ has positive definite Hermitian part. With this, the Nevanlinna--Pick interpolation theory in Section~\ref{sec:interpolation} can be converted to interpolating strictly positive real functions.

With the notion of positive realness, the ``$\rho$-stability'' version of the circle criterion is stated below.

\begin{lemma}\label{lem:circle}
    Consider a Lur'e system where the memoryless nonlinearity $\Delta:\IR^d\to\IR^d$ is sector-bounded in $[0,\infty)$ and let $\Psi(z)$ denote the $d\times d$ transfer function matrix of the LTI subsystem. The Lur'e system is $\rho$-stable if the transfer function matrix $\Psi(\gamma z)$ is strictly positive real for all $\gamma\in(\rho,1)$.
\end{lemma}

\begin{proof}
    It follows using the approach in \cite{Zhang2022convergence}.
\end{proof}

\begin{remark}
     As we are only concerned with asymptotic convergence rates (Definition~\ref{def:converge} and Remark~\ref{rmk:converge}), one may conjecture that, whenever the Lur'e system is stable for all $\Delta$'s sector-bounded in $[k_1,k_2]$, then
     the worst-case convergence rate of the Lur'e system for
     such $\Delta$'s is
    equal to
    that where $\Delta$ is replaced by a static gain $k\in[k_1,k_2]$. While this seems plausible since the stability guarantees ensure that the system state converges to the origin where asymptotic performance may depend on linearized dynamics, examples given in \cite{Zhang2022convergence,wang_numerical_2021} invalidate such a conjecture.
    \null\hfill$\Diamond$
\end{remark}

Following the typical approach in Lur'e systems, we apply a \emph{loop transformation} \cite{khalil2002nonlinear} on $\Delta_f$ to convert this into sector-bounded in $[0,\infty)$, which accordingly transforms $G(z)$ into
\begin{equation*}
    \Psi(z) = {\left(1+\mu G(z)\right)}^{-1} \left(1+\ell G(z)\right).
\end{equation*}
The inverse loop transformation is
\begin{equation}\label{eq:inv_loop}
    G(z) = {(\ell-\mu \Psi(z))}^{-1} \left(\Psi(z)-1\right).
\end{equation}
Then, from Lemma~\ref{lem:circle}, the system in Fig.~\ref{fig:alg_feedback} is $\rho$-stable provided that
\begin{equation}\label{eq:positive_psi}
    \Psi(\gamma z) \text{ is strictly positive real for all } \gamma\in(\rho,1).
\end{equation}
Moreover, Proposition~\ref{prop:LTI} requires that $G(z)$ have a pole at $z=1$, i.e.,
\begin{equation}\label{eq:inter_1}
    \Psi(1) = \kappa = \ell/\mu.
\end{equation}
Denote  $G(\infty)=\alpha$ the direct feedthrough gain. Then
\begin{equation}\label{eq:inter_inf}
    \Psi(\infty) = \frac{1+\ell \alpha}{1+\mu \alpha}.
\end{equation}
Utilizing the conformal map \eqref{eq:conformal}, the following theorem readily follows from the interpolation theory in Section~\ref{sec:interpolation}.
\begin{theorem}\label{thm:circle_bi}
    There exists a function $\Psi(z)$ that satisfies conditions \eqref{eq:positive_psi}, \eqref{eq:inter_1}, and \eqref{eq:inter_inf},
    if and only if
    \begin{equation}\label{eq:rho_bi}
        \rho \geq \rho_{\mathrm{C}} \coloneqq \frac{\kappa - 1}{\kappa + 1 + 2\ell \alpha}.
    \end{equation}
For $\rho=\rho_{\mathrm{C}}$, the function $\Psi(z)$ is unique and is given by
\begin{equation}\label{eq:Psi_bi}
    \Psi(z) = \kappa \frac{\left( 1-\rho \right) \left( z + \rho \right)}{\left( 1+\rho \right) \left( z - \rho \right)}.
\end{equation}
\end{theorem}
\vspace{1ex}

For a given direct feedthrough gain $\alpha$, Theorem~\ref{thm:circle_bi} establishes the fastest convergence rate ascertained by the circle criterion (Lemma~\ref{lem:circle}) for functions in $\mathcal{S}_{\mu,\ell}$.
Conversely, for a specified convergence rate $\rho\in(0,1)$, it provides the minimum direct feedthrough gain
\begin{equation}\label{eq:alpha_rho}
     \alpha_{\rho} \coloneqq  \frac{\left(1-\rho\right) \kappa - \left(1+\rho\right)}{2\rho \ell}.
\end{equation}
The transfer function $G(z)$ that achieves the $\rho$-convergence with $G(\infty)=\alpha_{\rho}$ can be synthesized by substituting \eqref{eq:Psi_bi} into \eqref{eq:inv_loop}, yielding
\begin{equation}\label{eq:tf_bi_gd}
    G(z) = \frac{\alpha_\rho z + \beta}{z - 1},
\end{equation}
where
\begin{equation*}
    \beta = \frac{2 + \alpha_\rho \left(\ell+\mu\right)}{\ell + \mu + 2\mu \ell \alpha_\rho} .
\end{equation*}
The corresponding implicit time-domain iteration is
\begin{equation}\label{eq:bi_gd_iter}
    x[t+1] = x[t] - \beta \nabla f(x[t]) - \alpha_\rho \nabla f(x[t+1]).
\end{equation}

When $\alpha=0$, the fastest rate $\rho_{\mathrm{C}}$ in \eqref{eq:rho_bi} reduces to
\begin{equation}\label{eq:rho_gd}
    \rho_{\mathrm{GD}} \coloneqq \frac{\kappa-1}{\kappa+1} = \frac{\ell-\mu}{\ell+\mu}.
\end{equation}
This rate coincides with the proven rate of the gradient descent method with optimal step size discovered by Polyak in \cite{Polyak1963} for all functions in $\mathcal{F}_{\mu,\ell}$ (see also \cite[Theorem 2.1.15]{nesterov2018lectures}), namely,
\begin{equation}\label{eq:gd}
    x[t+1] = x[t] - \frac{2}{\mu+\ell} \nabla f(x[t]),
\end{equation}
which follows from \eqref{eq:bi_gd_iter} by letting $\rho=\rho_{\mathrm{GD}}$.
Note that the gradient descent \eqref{eq:gd} is the fastest strictly causal algorithm ascertained by the circle criterion for all functions in $\mathcal{S}_{\mu,\ell}$, which might suggest that algorithms with momentum, such as Polyak's heavy-ball method \eqref{eq:heavy_ball}, cannot accelerate on functions in $\mathcal{S}_{\mu,\ell}$.

\begin{remark}
In various previous works, larger classes of functions are considered, including, e.g., those that satisfy respectively the inequalities
\begin{subequations}\label{eq:sector123}
    \begin{align}
        \ell {\left\| x-x^{\star} \right\|} & \geq {\left\| \nabla f(x) \right\|} \geq \mu {\left\| x-x^{\star} \right\|}, \label{eq:sector1}\\
        {\left< \nabla f(x), x-x^{\star} \right>} &\geq \mu {\left\| x-x^{\star} \right\|}^{2} , \label{eq:sector2}\\
        \tfrac{1}{2} {\left\| \nabla f(x) \right\|}^{2} &\geq \tfrac{\mu^{2}}{\ell} {\bigl( f(x)-f(x^{\star}) \bigr)}. \label{eq:sector3}
    \end{align}
\end{subequations}
Here the inequalities in \eqref{eq:sector1} are called \emph{error bounds} \cite{guille-escuret_gradient_2022}, \eqref{eq:sector2} is known as \emph{restricted secant inequality} (RSI) \cite{guille-escuret_gradient_2022},
and \eqref{eq:sector3} the \emph{Polyak--\L ojasiewicz (PL) inequality} with constant ${\mu^{2}}/{\ell}$. It can be readily shown that any $f\in\mathcal{S}_{\mu,\ell}$ satisfies these three inequalities, and further, the PL inequality defines the largest of the four function sets including $\mathcal{S}_{\mu,\ell}$.
It is known \cite{guille-escuret_gradient_2022} that gradient descent is optimal for functions satisfying \eqref{eq:sector1} and \eqref{eq:sector2}.
This, along with the fact that gradient descent \eqref{eq:gd} is optimal among all algorithms achievable by the circle criterion, suggests that gradient descent may indeed be the fastest for all four function classes, despite their varying generality and characteristics (e.g., functions satisfying \eqref{eq:sector1} may not be Lipchitz smooth, and those satisfying \eqref{eq:sector2} and \eqref{eq:sector3} need not be convex).
\null\hfill$\Diamond$
\end{remark}


The convergence rate $\rho_\mathrm{C}$ can be made arbitrarily fast by increasing $\alpha>0$ sufficiently. However, unlike the case of quadratic functions in Section~\ref{sec:bi_quadratic}, implicit algorithms for non-quadratic functions cannot be easily translated into an implementable causal form, since \eqref{eq:bi_gd_iter} involves solving a nonlinear equation for $x[t+1]$.
Alternatively, the problem
\begin{equation*}
    x[t+1] = \argmin_{\xi\in\IR^d} \biggl\{\alpha_\rho f(\xi) + \frac{1}{2}{\Bigl\| \xi - \bigl(x[t] - \beta \nabla f(x[t]) \bigr) \Bigr\|}^2 \biggr\},
\end{equation*}
solves the equation, and can be rewritten concisely as
\begin{equation}\label{eq:bi_gd_prox}
    x[t+1] = \mathrm{prox}_{\alpha_\rho f} \bigl(x[t] - \beta \nabla f(x[t]) \bigr),
\end{equation}
using the proximal operator $\mathrm{prox}_f\colon\IR^d\to\IR^d$ defined by \cite{parikh_proximal_2014}:
\begin{equation*}
    \mathrm{prox}_f(x) \coloneqq \argmin_{\xi\in\IR^d} \left\{f(\xi) + \frac{1}{2}{\|\xi-x\|}^2\right\}.
\end{equation*}

Each iteration in \eqref{eq:bi_gd_prox} involves solving a causal optimization problem, so the algorithm would be especially useful when the proximal operator is computationally easier to evaluate; the practical importance of proximal algorithms is widely known for its broad applications to, e.g., optimization of nonsmooth functions (see, e.g., \cite{parikh_proximal_2014, taylor_exact_2018}), as also evidenced by our development in the forthcoming section. Similar to the application of implicit algorithms for quadratic functions in Section~\ref{sec:bi_ill_quadratic}, \eqref{eq:bi_gd_prox} can accelerate convergence in problems with ill-conditioned functions, in the sense that $\kappa=\ell/\mu$ is large. Indeed, in such a case, direct optimization via gradient-based methods
converges slowly. On the other hand,  \eqref{eq:bi_gd_prox} can achieve a fast convergence rate $\rho$ for a suitably large $\alpha_\rho > 0$. Evaluating the proximal operator may be implemented via gradient descent, which now requires only few iterations as
\begin{equation}\label{eq:k_sub}
    \kappa_{\mathrm{sub}} =  \frac{\ell_{\mathrm{sub}}}{\mu_{\mathrm{sub}}} = \frac{1+\alpha_\rho \ell}{1+\alpha_\rho \mu},
\end{equation}
the condition number of the sub-problem, can be made significantly small
with a choice of $\alpha_\rho$.
This results in a tradeoff where a faster convergence rate for the main iteration requires a larger $\alpha_\rho$, while a faster convergence rate for solving the sub-problem requires a smaller $\alpha_\rho$. The optimal choice of $\alpha_\rho$ can be determined by balancing these two rates. A detailed analysis on this tradeoff is beyond the scope of this work.

The algorithm \eqref{eq:bi_gd_prox} can be generalized to optimize functions $f\in\mathcal{S}_{\mu,\infty}$.
For $\ell\to\infty$, $\rho_\mathrm{C}$ in \eqref{eq:rho_bi} becomes
$\rho_\mathrm{C} = \frac{1}{1 + 2\mu\alpha}$.
Accordingly, to achieve a rate $\rho$, the minimum direct feedthrough gain from \eqref{eq:alpha_rho} is $\alpha_\rho = \frac{1-\rho}{2\rho\mu}$.
Then, \eqref{eq:tf_bi_gd} gives $G(z) = \alpha_\rho \frac{z + \rho}{z - 1}$,
resulting in the proximal algorithm
\begin{equation}\label{eq:bi_gd_prox_mu}
    x[t+1] = \mathrm{prox}_{\alpha_\rho f} \bigl(x[t] - \alpha_\rho \rho \nabla f(x[t]) \bigr).
\end{equation}




\subsection{Algorithms with splitting}

Implicit algorithms for non-quadratic functions can be converted into a causal form with the help of proximal operators.
However, in many cases, evaluating the proximal operator may be as difficult as the original problem.
Nevertheless, proximal operators do provide appealing advantages when the cost function $f$
can be split into the form
$f(x) = h(x) + g(x)$
with $g$ proximable (i.e., its proximal operator is computationally efficient to evaluate) \cite{parikh_proximal_2014}. Intuitively, the idea is to, where appropriate, split a function into a ``well-behaving'' part and a more problematic part which nonetheless is proximable. This generalizes the ideas of implicit algorithms for gain margin improvement to considerably broader classes of functions. We demonstrate below how an efficient algorithm can be devised using such splitting.

We assume that $h\in\mathcal{F}_{\mu_{1},\ell_{1}}$ and $g\in\mathcal{F}_{\mu_{2},\ell_{2}}$.
%
Just as in algorithms without splitting $f$, an LTI algorithm with splitting can be cast as a feedback system as in Fig.~\ref{fig:prox_feedback}, where $G(z)$ is the $2\times 2$ transfer function matrix that characterizes the algorithm. In this, $\Delta_h(\cdot)$ is defined as:
\begin{equation*}
    \Delta_h(e_1[t]) \coloneqq \nabla h(x^\star) - \nabla h(x^\star-e_1[t]) = \nabla h(x^\star) - \nabla h(x_1[t]),
\end{equation*}
and similarly for $\Delta_g(\cdot)$. Since, $x^\star$ is the minimizer of $f$,
\begin{equation*}
    \nabla f(x^\star) = \nabla h(x^\star) + \nabla g(x^\star) = \mathbf{0}_d,
\end{equation*}
and thus
\begin{equation}\label{eq:split_grad}
    \begin{bmatrix}
        \nabla h(x^\star) \\ \nabla g(x^\star)
      \end{bmatrix} =  \begin{bmatrix}
        1 \\ -1
    \end{bmatrix} \otimes \nabla h(x^\star),
\end{equation}
where typically $\nabla h(x^\star) \neq \mathbf{0}_d$. This explains the exogenous signal $\begin{bsmallmatrix}
        \nabla h({x}^\star) \\ \nabla g({x}^\star)
      \end{bsmallmatrix}$
in Fig.~\ref{fig:prox_feedback} which is absent in the corresponding Fig.~\ref{fig:alg_feedback} and Fig.~\ref{fig:multi_grad} for algorithms without splitting.

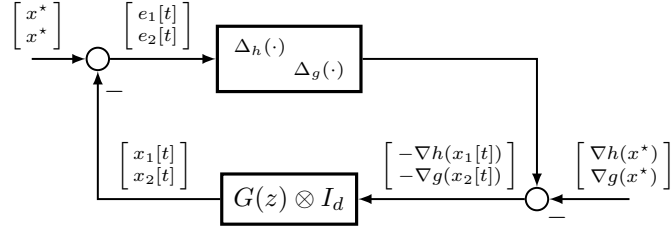
\begin{figure}[tb]
    \centering
    \begin{tikzpicture}[>=latex, line width=0.75pt]
    \tikzstyle{component} = [draw, rectangle, very thick, inner sep=1ex, minimum height=1em]
     \tikzstyle{add} = [draw, circle, inner sep=3pt]

     \node[component] (g) {${G}(z)\otimes I_d$};
     \node[component, above = 3em of g] (d) {$\begin{smallmatrix}
      \Delta_h(\cdot) &  \\
        & \Delta_g(\cdot) \\
     \end{smallmatrix}$};

    \node(a1) [add, right=6.3em of g]  {};
      \node(a2) [add, left=4em of d] {};

    \draw[->] (a2) --node[above,pos=0.45]{$\begin{bsmallmatrix}
      e_1[t] \\ e_2[t]
    \end{bsmallmatrix}$} (d);
    \draw[->] (d) -| (a1);

    \draw[->]  (a1) -- node[above,pos=0.45]{$\begin{bsmallmatrix}
      -\nabla h(x_1[t]) \\  -\nabla g(x_2[t])
    \end{bsmallmatrix}$} (g);

    \draw[->] (g) -| node[above,pos=0.27]{$\begin{bsmallmatrix}
      x_1[t] \\ x_2[t]
    \end{bsmallmatrix}$} node[pos=0.92,right=-0.3ex]{$-$} (a2);

    \draw[<-] (a2) -- node[pos=0.9,above=0.73ex, inner sep=0]{$\begin{bsmallmatrix}
      {x}^\star \vphantom{e_1[t]} \\ {x}^\star \vphantom{e_1[t]}
    \end{bsmallmatrix}$} ++(-2.5em,0);

    \draw[<-] (a1) -- node[pos=0.1, below]{$-$} node[pos=0.93,above]{$\begin{bsmallmatrix}
        \nabla h({x}^\star) \\ \nabla g({x}^\star)
      \end{bsmallmatrix}$} ++(3.5em,0);
\end{tikzpicture}
    \caption{Algorithm with splitting as a feedback system.}
    \label{fig:prox_feedback}
\end{figure}

It is clear from Fig.~\ref{fig:prox_feedback} that $G(z)$ must have an accumulator.
Consider the coprime factorization of $G(z)$ as $G(z) = {M(z)}^{-1}N(z)$, where $M(z),N(z)$ are coprime $2\times 2$ transfer function matrices with entries in $\mathcal{RH}_\infty$. Analogous to (\ref{eq:Minv_1}),
this leads to a necessary condition for the convergence of the algorithm characterized by $G(z)$.

\begin{theorem}
    For any $x^\star\in\IR^d$ that satisfies \eqref{eq:split_grad}, the algorithm characterized by the $2\times 2$ transfer function matrix $G(z)$ converges to the solution $x^\star$ only if $G(z)$ satisfies the interpolation conditions:
    \begin{equation}\label{eq:G_1}
        M(1) \begin{bmatrix}
            1 \\ 1
        \end{bmatrix} = \begin{bmatrix}
            0 \\ 0
        \end{bmatrix}, \quad\text{and}\quad N(1) \begin{bmatrix}
            1 \\ -1
        \end{bmatrix} = \begin{bmatrix}
            0 \\ 0
        \end{bmatrix}.
    \end{equation}
\end{theorem}
\vspace{1ex}
\begin{proof}
   The feedback equation for Fig.~\ref{fig:prox_feedback} in the $\mathcal{Z}$-domain reads
   \begin{equation*}
       \tfrac{z}{z-1}\begin{bsmallmatrix}
          x^\star \vphantom{\nabla h({x}^\star)} \\
          x^\star \vphantom{\nabla h({x}^\star)}
      \end{bsmallmatrix} - \left(G(z) \otimes I_d\right) \left( \widehat{\Delta(\mathbf{e})} (z) - \tfrac{z}{z-1}
      \begin{bsmallmatrix}
          \nabla h({x}^\star) \\ \nabla g({x}^\star)
        \end{bsmallmatrix}\right)
        = \widehat{\mathbf{e}} (z).
  \end{equation*}
   where $\widehat{\mathbf{e}} (z)$ and $\widehat{\Delta_f(\mathbf{e})} (z)$ respectively denote the $\mathcal{Z}$-transform of the input and output of the nonlinear component, i.e.,
   \begin{equation*}
       \widehat{\Delta(\mathbf{e})} (z) \coloneqq \begin{bmatrix}
           \widehat{\Delta_h(e_1)}(z) \\ \widehat{\Delta_g(e_2)}(z)
       \end{bmatrix},\quad
       \widehat{\mathbf{e}} (z) \coloneqq
       \begin{bmatrix}
           \widehat{e}_1(z) \vphantom{\widehat{\Delta_h(e_1)}(z)}  \\
           \widehat{e}_2(z) \vphantom{\widehat{\Delta_h(e_1)}(z)}
       \end{bmatrix}.
   \end{equation*}
       Pre-multiplying both sides by $\frac{z-1}{z}{{M}(z)}\otimes I_d$ and substituting \eqref{eq:split_grad}, we get
    \begin{equation*}
        \bigl({{M}(z)} \otimes I_d\bigr)
      \bigl( \begin{bsmallmatrix}
       1 \\ 1
      \end{bsmallmatrix} \otimes {x}^\star \bigr) +
      \bigl({{N}(z)} \otimes I_d\bigr)
      \bigl( \begin{bsmallmatrix}
       1 \\ -1
      \end{bsmallmatrix} \otimes \nabla h({x}^\star) \bigr)
        =
       \frac{z-1}{z}
      \Bigl( \bigl({N}(z)\otimes I_d\bigr)
      \widehat{\Delta(\mathbf{e})} (z) +
      \bigl({{M}(z)}\otimes I_d\bigr)
     \widehat{\mathbf{e}} (z) \Bigr).
        \end{equation*}
    Since ${M}(z), {N}(z)\in\mathcal{RH}_\infty$,
    if ${e}_1[t]$ and ${e}_2[t]$ converges to zero, then the right-hand side of the above equation tends to 
    \begin{equation*}
       \mathbf{0}_{2d} = \left( {M(1)} \begin{bsmallmatrix}
           1 \\ 1
          \end{bsmallmatrix}  \right) \otimes x^\star + \left( {N(1)} \begin{bsmallmatrix}
           1 \\ -1
             \end{bsmallmatrix}  \right) \otimes \nabla h({x}^\star)
    \end{equation*}
    holds for all ${x}^\star,\nabla h({x}^\star)\in\IR^d$, which establishes \eqref{eq:G_1}.
    \end{proof}

The interpolation conditions \eqref{eq:G_1} can be interpreted from the perspective of two standard control problems: tracking and regulation, among which the former leads to the singularity of $M(1)$ for tracking a step reference with direction $\begin{bsmallmatrix}
    1 \\ 1
\end{bsmallmatrix}$, and the latter leads to the singularity of $N(1)$ for regulating a step disturbance with direction $\begin{bsmallmatrix}
    1 \\ -1
\end{bsmallmatrix}$.

As we assume that $g$ is proximable, $G(z)$ can be not strictly causal, and specifically $G(\infty)$ can take the form:
\begin{equation}\label{eq:G_infty}
    G(\infty) = \begin{bmatrix}
        0 & \eta_1 \\ 0 & \eta_2
    \end{bmatrix},
\end{equation}%
where $\eta_1$ and $\eta_2$ are tunable parameters.

Since $h\in\mathcal{F}_{\mu_{1},\ell_{1}}$, its gradient difference $\Delta_h(\cdot)$ is slope-restricted in $[\mu_{1},\ell_{1}]$, and for the same reason, $\Delta_g(\cdot)$ is slope-restricted in $[\mu_{2},\ell_{2}]$.
We now utilize a loop transformation to convert $\Delta_h(\cdot)$ and $\Delta_g(\cdot)$ to nonlinearities slope-restricted in $[0,\infty)$, which accordingly transform $G(z)$ to
\begin{equation*}
        \Psi(z) = {\left(I_{2} + \begin{bmatrix}
            \ell_{1} & {0} \\ {0} & \ell_{2}
        \end{bmatrix} G(z)\right)} {\left(I_{2} + \begin{bmatrix}
            \mu_{1} & {0} \\ {0} & \mu_{2}
        \end{bmatrix} G(z)\right)\!\!}^{-1}.
\end{equation*}
The inverse transformation is
\begin{equation}\label{eq:inv_loop2}
    \setlength{\arraycolsep}{0.9ex}
    \hspace{-0.7em}   G(z) = {\biggl( \begin{bmatrix}
            {\ell}_{1} & {0} \\ {0} & \ell_{2}
            \end{bmatrix} - \Psi(z) \begin{bmatrix}
            {\mu}_{1} & {0} \\ {0} & \mu_{2}
            \end{bmatrix} \biggr)\!\!}^{-1} 
            \biggl(\Psi(z) - I_{2} \biggr). \hspace{-0.2em}
\end{equation}
The interpolation conditions \eqref{eq:G_1} and \eqref{eq:G_infty} on $G(z)$ can be converted to those on $\Psi(z)$ as
    \begin{equation}\label{eq:Psi_inter}
    \Psi(1) = \begin{bmatrix}
         \frac{\ell_{1}+\mu_{2}}{\mu_{1}+\mu_{2}}
        & \frac{\ell_{1}-\mu_{1}}{\mu_{1}+\mu_{2}} \\
        \frac{\ell_{2}-\mu_{2}}{\mu_{1}+\mu_{2}} & \frac{\ell_{2}+\mu_{1}}{\mu_{1}+\mu_{2}}
    \end{bmatrix}, \quad
    \Psi(\infty) =  \begin{bmatrix}
        1 &
        \frac{\ell_1-\mu_1}{1+\mu_{2} \eta_{2}} \eta_1 \\ 0 &
        \frac{1+\ell_{2} \eta_{2}}{1+\mu_{2} \eta_{2}}
    \end{bmatrix}.
\end{equation}
Note that after loop transformation, the nonlinear component still possesses the diagonal structure as in Fig.~\ref{fig:prox_feedback}. We may abandon this structure and directly employ the circle criterion, which, however, suffers from certain conservatism. The following lemma, inspired by the analysis of \emph{structured singular value} \cite{zhou1998essentials} in robust control, paves the way to a less conservative result.

\begin{lemma}\label{lem:sector}
    If $\Delta_i\colon \IR^{d_i} \to \IR^{d_i}$ is sector-bounded in $[k_1,k_2]$, $i=1,\dots,n$, then $\Delta\colon \IR^{d_1+\dots+d_n} \to \IR^{d_1+\dots+d_n}$ defined by
    \begin{equation*}
        \Delta = \begin{bsmallmatrix}
            w_1 I_{d_1} &  & \\  & \ddots & \\
            & & w_n I_{d_n}
        \end{bsmallmatrix} \circ
        \begin{bsmallmatrix}
            \Delta_1(\cdot) &  & \\  & \ddots & \\
            & & \Delta_n(\cdot)
        \end{bsmallmatrix} \circ
        \begin{bsmallmatrix}
            \frac{1}{w_1} I_{d_1} & & \\  & \ddots & \\
            & & \frac{1}{w_n} I_{d_n}
        \end{bsmallmatrix}
    \end{equation*}
    is sector-bounded in $[k_1,k_2]$ for any $w_i \neq 0$, $i=1,\dots,n$.
\end{lemma}
\vspace{1ex}
\begin{proof}
Since $\Delta_i$ is sector-bounded in $[0,\infty)$, then for any $w_i \neq 0$ and for all $x_i\in\IR^{d_i}$, $i=1,\dots,n$, we have
\begin{equation*}
   \bigl<\Delta_i (\tfrac{1}{w_i} x_i) , \tfrac{1}{w_i} x_i \bigr> \geq 0 .
\end{equation*}
By the definition of $\Delta$, we have
\begin{align*}
   \hspace{-1em}
   \left<\Delta \Bigl(\begin{bsmallmatrix}
       x_1 \\ \vdots \\ x_n
   \end{bsmallmatrix} \Bigr) , \begin{bsmallmatrix}
   x_1 \\ \vdots \\ x_n
\end{bsmallmatrix} \right> &= \left<\begin{bsmallmatrix}
   w_1 \Delta_1(\tfrac{1}{w_1} x_1) \\ \vdots \\
   w_n \Delta_n(\tfrac{1}{w_n} x_n)
\end{bsmallmatrix} , \begin{bsmallmatrix}
   x_1 \\ \vdots \\ x_n
\end{bsmallmatrix} \right> \\
&= \sum_{i=1}^{n} w_i^2 \bigl<\Delta_i (\tfrac{1}{w_i} x_i) , \tfrac{1}{w_i} x_i \bigr> \geq 0.
\end{align*}
This completes the proof.
\end{proof}

From Lemma~\ref{lem:sector} and the circle criterion (Lemma~\ref{lem:circle}), the Lur'e system in Fig.~\ref{fig:prox_feedback} is $\rho$-stable if
\begin{equation}\label{eq:positive_psi_prox}
    \text{there exists } w\neq 0 \text{ such that }
    \begin{bmatrix}
        w & 0 \\ 0 & 1
    \end{bmatrix} \Psi(\gamma z)
    \begin{bmatrix}
        1/{w} & 0 \\ 0 & 1
    \end{bmatrix}
    \text{is strictly positive real for all } \gamma \in (\rho,1).
\end{equation}

\begin{theorem}\label{thm:circle_prox}
    There exists a function $\Psi(z)$ that satisfies \eqref{eq:Psi_inter} and \eqref{eq:positive_psi_prox}, only if
    \begin{equation}\label{eq:rho_s}
        \rho \geq \rho_{\mathrm{S}} \coloneqq \frac{\ell_{1}-\mu_{1}}{\ell_{1}+\mu_{1} + 2\mu_{2}}.
    \end{equation}
    Moreover, when $\rho = \rho_{\mathrm{S}}$, such a function $\Psi(z)$ exists if and only if
    \begin{equation*}
    w^2 = \frac{\ell_{2}-\mu_{2}}{{\ell_{1}-\mu_{1}} } \rho_{\mathrm{S}}^{-1}.
    \end{equation*}
    %
    One 
    such function is given by
    \begin{equation}\label{eq:Psi_prox}
        \Psi(z) = \frac{1}{z - \rho_{\mathrm{S}}} \begin{bmatrix}
            z + \rho_{\mathrm{S}} & 2 \rho_{\mathrm{S}} z \\
            \frac{2 (\ell_2-\mu_2) }{\ell_{1}+\mu_{1} + 2\mu_{2}} &
            \frac{\ell_1+\mu_1+2\ell_2}{\ell_{1}+\mu_{1} + 2\mu_{2}} z - \rho_{\mathrm{S}}
        \end{bmatrix}.
    \end{equation}
\end{theorem}
\vspace{1ex}
\begin{proof}
    See Appendix B.
\end{proof}

Theorem~\ref{thm:circle_prox} provides the fastest convergence rate $\rho_{\mathrm{S}}$ of LTI algorithms with splitting that can be established using the circle criterion.
The transfer function matrix $G(z)$ that achieves the fastest convergence rate $\rho_{\mathrm{S}}$ can be synthesized by substituting \eqref{eq:Psi_prox} into \eqref{eq:inv_loop2}, yielding
\begin{equation*}
    G(z) = \frac{2}{\mu_1+\ell_1} \frac{1}{z-1}
    \begin{bmatrix}
        1 \\ 1
    \end{bmatrix}
    \begin{bmatrix}
        1 & z
    \end{bmatrix}.
\end{equation*}
Its minimal realization corresponds to the proximal gradient method \cite{parikh_proximal_2014, taylor_exact_2018} given by
\begin{equation}\label{eq:gd_prox}
    x[t+1] = \mathrm{prox}_{\frac{2}{\mu_1+\ell_1} g} \bigl(x[t] - \tfrac{2}{\mu_1+\ell_1} \nabla h(x[t]) \bigr).
\end{equation}
Note that $G(z)$ satisfies the interpolation condition (\ref{eq:G_infty}) with $\eta_1=\eta_2=\tfrac{2}{\mu_1+\ell_1}$.

Common assumptions on the proximable part $g$ are that it is convex, closed, and proper (see, e.g. \cite{parikh_proximal_2014, taylor_exact_2018}), which implies $\mu_{2}=0$ and $\ell_{2}=\infty$. In this case, $\rho_{\mathrm{S}}=\frac{\ell_{1}-\mu_{1}}{\ell_{1}+\mu_{1}}$, which is the fastest convergence rate of the proximal gradient method \cite{taylor_exact_2018}, and \eqref{eq:gd_prox} is the proximal gradient method with the optimal step size $\frac{2}{\mu_1+\ell_1}$.
Our synthesis of the proximal algorithm \eqref{eq:gd_prox} shows that if $g$ is further strongly convex, then the convergence rate can be improved without changing the step size. Furthermore, it also shows that the Lipschitz smoothness of $g$ does not improve the convergence rate.




\subsection{Simulation Examples}

\subsubsection{Implicit Algorithm}
To verify the implicit algorithm \eqref{eq:bi_gd_prox}, we consider the cost function $h\colon \IR^d\to\IR$ given by
\begin{equation*}
    h(x) = \sum_{i=1}^d \phi(a_i\TP x - b_i)  \;\text{where}\;
    \phi(\nu) = \begin{cases}
        \frac{\ell}{2} \nu^2 & \text{if } \nu \geq 0, \\
        \frac{\mu}{2} \nu^2 & \text{if } \nu < 0.
    \end{cases}
\end{equation*}
Here, $a_i,b\in\IR^d$ are chosen randomly so that $A=\bigl[
    a_1 \, \cdots \, a_d \bigr] \in\IR^{d\times d}$ is unitary, and thus $h\in\mathcal{F}_{\mu,\ell}$. It is easy to see that the minimizer of $h$ is $x^\star=A b$. We set $d=100$, $\mu=0.01$, and $\ell=100$. The $\mathrm{prox}_{\alpha_\rho h}(x)$ is computed using gradient descent \eqref{eq:gd} with $\mu_{\mathrm{sub}}$ and $\ell_{\mathrm{sub}}$ given in \eqref{eq:k_sub}, until the iterative variable $\xi[k]$ satisfies $\left\|\xi[k]-\xi[k-1]\right\|\leq 0.01\left\|x\right\|$ with $\xi[-1]=\mathbf{0}_d$ and $\xi[0]=x$.


For a fair comparison of the effectiveness of the implicit algorithm \eqref{eq:bi_gd_prox} for different $\alpha_\rho$, we count the number of times the gradient $\nabla h$ is computed, including those computed for solving $\mathrm{prox}_{\alpha_\rho h}(x)$, until the error $\left\|x[k]-x^\star\right\|\leq {10}^{-10}$ with initial $x[0]=\mathbf{0}$.
    The simulation results are shown in Fig.~\ref{fig:simulation_bi}.
    For this ill-conditioned function $h$, the gradient descent method ($\alpha_\rho=0$) requires $1.38\times 10^5$ times of computing the gradient $\nabla h$ to achieve the specified accuracy. In contrast, the implicit algorithm \eqref{eq:bi_gd_prox} with $\alpha_\rho=10$ requires only $784$ times of computing $\nabla h$ to achieve the same accuracy. This demonstrates the effectiveness of the implicit algorithm in reducing computational complexity.

    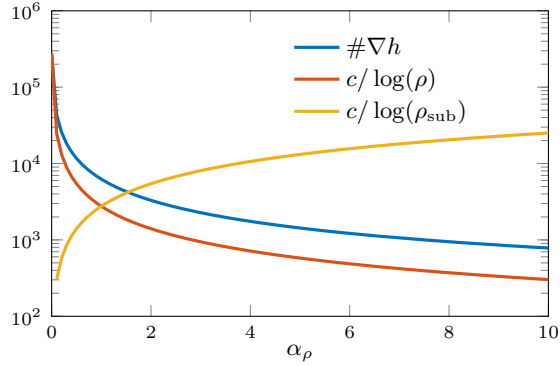
\begin{figure}[tb]
        \centering
        \definecolor{mycolor1}{rgb}{0.00000,0.44700,0.74100}%
\definecolor{mycolor2}{rgb}{0.85000,0.32500,0.09800}%
\definecolor{mycolor3}{rgb}{0.92900,0.69400,0.12500}%
\begin{tikzpicture}

\begin{axis}[%
width=0.6\textwidth,
height=1.2in,
at={(1.053in,0.668in)},
scale only axis,
xmin=0,
xmax=10,
ymode=log,
ymin=600,
ymax=200000,
yminorticks=true,
xlabel style={font=\footnotesize, yshift=0.9ex},
scaled ticks = false,
ticklabel style = {font=\scriptsize},
xlabel={$\alpha_\rho$},
legend pos=north east,
legend style={draw=none, fill=none, yshift=-0.5ex, xshift=-2em, font=\footnotesize},
legend cell align={left}
]
\addplot [color=mycolor1, very thick]
  table[row sep=crcr]{%
0	138001\\
0.1010101010101	43310.0000000001\\
0.202020202020202	25398\\
0.303030303030303	18106\\
0.404040404040405	14114\\
0.505050505050505	11588\\
0.606060606060606	9842\\
0.707070707070708	8564.00000000002\\
0.808080808080808	7584.00000000001\\
0.909090909090908	6811.99999999999\\
1.01010101010101	6185.99999999999\\
1.11111111111111	5666\\
1.21212121212121	5232.00000000001\\
1.31313131313131	4860.00000000001\\
1.41414141414141	4538.00000000001\\
1.51515151515152	4260\\
1.61616161616162	4014\\
1.71717171717172	3794\\
1.81818181818182	3600\\
2.02020202020202	3266\\
2.22222222222222	2992.00000000001\\
2.42424242424242	2762\\
2.62626262626263	2566\\
2.92929292929293	2322.00000000001\\
3.23232323232323	2122\\
3.53535353535354	1956\\
3.73737373737374	1860\\
4.14141414141414	1694\\
4.54545454545455	1558\\
4.94949494949495	1444\\
5.35353535353535	1346\\
5.85858585858586	1242\\
6.56565656565657	1124\\
7.57575757575758	992.000000000002\\
7.87878787878788	960.000000000001\\
8.08080808080808	938.000000000001\\
8.48484848484848	900.000000000001\\
8.78787878787879	874\\
8.98989898989899	858.000000000001\\
9.19191919191919	840.000000000001\\
9.39393939393939	826.000000000001\\
9.7979797979798	796\\
10	784.000000000001\\
};
\addlegendentry{$\# \nabla h$}

\end{axis}
\end{tikzpicture}%
        \caption{Simulation results of implicit algorithm \eqref{eq:bi_gd_prox}, where $\# \nabla h$ is the number of times the gradient $\nabla h$ is computed (counting sub-iterations).}
        %
        \label{fig:simulation_bi}
    \end{figure}

    \subsubsection{Proximal Algorithm}
To verify the proximal algorithm \eqref{eq:gd_prox}, we consider the cost function $f(x) = h(x) + g(x)$ with $h$ identical to the previous example and
\begin{equation*}
    g(x) = \begin{cases}
        0.5 \sum_{i=1}^d w_i x_i^2 & \text{if } x_i \geq 0 \; \forall i=1,\dots,d, \\
        +\infty & \text{otherwise}.
    \end{cases}
\end{equation*}
While non-differentiable, $g$ is $\min(w)$-strongly convex and proximable with proximal operator being
\begin{equation*}
    {\Bigl(\mathrm{prox}_{\lambda g}(x)\Bigr)\!}_i =
    \max \Bigl\{\frac{x_{i}}{1+\lambda w_i}, 0 \Bigr\} , \; i=1,\dots,d.
\end{equation*}
We set $d=10^3$, $\ell=100$, $\mu=0.1$, and $\min(w)=1$. The simulation results are shown in Fig.~\ref{fig:simulation}, which clearly demonstrates that the error $\left\|e[t]\right\|$ and gradient norm $\min_{\xi\in\partial g(x[t])}\|\nabla h(x[t]) + \xi\|$ converge to zero at a rate not slower than $\rho_{\mathrm{S}}$ under different $\max(w)$. A counterintuitive observation is that larger $\max(w)$ leads to faster convergence in this setting, which, to some extent, validates the theoretical finding that the Lipschitz smoothness of the proximable part does not improve the convergence rate.

Even though our deduction to \eqref{eq:gd_prox} is based on the assumption that $g$ is differentiable, in this example, $g$ is not differentiable and $\Delta_g$ in Fig.~\ref{fig:prox_feedback} is defined with the subgradient $\partial g$ rendering the feedback system in Fig.~\ref{fig:prox_feedback} set-valued. Tools for analyzing set-valued Lur'e systems \cite{brogliato_dynamical_2020} promise a rigorous extension of the present frequency-domain framework for algorithms with splitting to the case where $g$ is not differentiable.

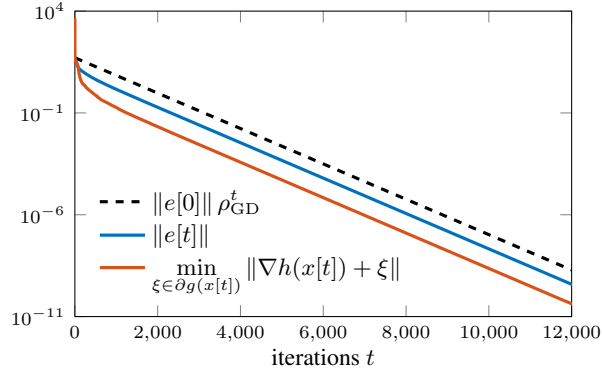
\begin{figure}[tb]
    \centering
 \definecolor{mycolor1}{rgb}{0.00000,0.447,1}%
\definecolor{mycolor2}{rgb}{1,0,1}%
\pgfplotsset{%
my legend/.style={legend image code/.code={%
\draw[mycolor1] (0ex,0.3ex) -- (3.8ex,0.3ex);
\draw[mycolor2] (0ex,-0.3ex) -- (3.8ex,-0.3ex);
}},%
}
\begin{tikzpicture}

\begin{axis}[%
width=0.6\textwidth,
height=1.4in,
scale only axis,
xmin=0,
xmax=1200,
ymode=log,
ymin=1e-11,
ymax=1e4,
yminorticks=true,
xlabel style={font=\footnotesize, yshift=0.9ex},
scaled ticks = false,
ticklabel style = {font=\scriptsize},
xlabel={iterations $t$},
legend pos=north east,
legend style={draw=none, fill=none, yshift=0.1ex, xshift=2ex, font=\footnotesize},
legend cell align={left}
]

\addplot [color=black, dotted, very thick]
  table[row sep=crcr]{%
0	45.4098296777094\\
1200	2.01764725015175e-10\\
};
\addlegendentry{$\left\|e[0]\right\|\rho_{\mathrm{S}}^t$}

\addlegendimage{my legend, very thick, solid}
\addlegendentry{$\left\|e[t]\right\|$}

\addlegendimage{my legend, very thick, dashed}
\addlegendentry{$\displaystyle \min_{\xi\in\partial g(x[t])}\|\nabla h(x[t]) + \xi\|$}

\addplot [color=mycolor1, very thick]
  table[row sep=crcr]{%
0	45.4098296777094\\
1	28.8630776946493\\
2	24.0578145425911\\
3	21.3111001097452\\
5	17.7879262894822\\
7	15.3356072811175\\
10	12.699141577748\\
14	10.2083329093175\\
18	8.41220174764508\\
24	6.45590127466307\\
32	4.63997466123567\\
41	3.28478385038315\\
53	2.12811575264331\\
67	1.30880775729063\\
87	0.67051830693068\\
107	0.352209742827247\\
126	0.19545966335253\\
153	0.086705728129294\\
184	0.0350504499370768\\
221	0.0121751313538617\\
255	0.00472031627121171\\
291	0.00177437690436644\\
330	0.000629943833321301\\
373	0.000206152017817578\\
420	6.23280853730465e-05\\
472	1.70061214995273e-05\\
530	4.09364736415902e-06\\
596	8.30094810957083e-07\\
671	1.38820334749167e-07\\
758	1.78765065146882e-08\\
868	1.37557208202859e-09\\
1051	1.99650254149168e-11\\
1052	1.92100589274425e-11\\
1053	1.9059308251908e-11\\
1054	1.83240890403475e-11\\
1055	1.81946342667051e-11\\
1056	1.74781689230463e-11\\
1057	1.73686599383695e-11\\
1058	1.66703872802199e-11\\
1059	1.65795971165816e-11\\
1060	1.58982154223179e-11\\
1061	1.58257201370516e-11\\
1062	1.51608514723923e-11\\
1063	1.51061686669695e-11\\
1064	1.44568532943827e-11\\
1065	1.44192546262994e-11\\
1066	1.37847980776622e-11\\
1067	1.37631232789747e-11\\
1068	1.31429147168624e-11\\
1069	1.31373787044126e-11\\
1070	1.25304908045057e-11\\
1071	1.25399086229107e-11\\
1072	1.19447674965918e-11\\
1073	1.19677369853978e-11\\
1074	1.13853234934992e-11\\
1075	1.1421173856234e-11\\
1076	1.0850958345103e-11\\
1077	1.09005128582604e-11\\
1078	1.03408395735547e-11\\
1079	1.04023230720895e-11\\
1080	9.85366351469617e-12\\
1081	9.92847787007407e-12\\
1082	9.38989615270784e-12\\
1083	9.47553858977372e-12\\
1084	8.94592355911929e-12\\
1085	9.04167786525173e-12\\
1086	8.51986700127826e-12\\
1087	8.62640212372934e-12\\
1088	8.11350117008597e-12\\
1089	8.2314720343587e-12\\
1090	7.72670844592445e-12\\
1091	7.85460528298006e-12\\
1092	7.35742477004528e-12\\
1093	7.49459867019744e-12\\
1094	7.00342541438878e-12\\
1095	7.14921489505682e-12\\
1096	6.66561325009116e-12\\
1097	6.82092800353306e-12\\
1098	6.34347939513263e-12\\
1099	6.50744546401658e-12\\
1100	6.03532633445133e-12\\
1101	6.20781350612627e-12\\
1102	5.74158211046326e-12\\
1103	5.92158221752472e-12\\
1104	5.46014315387739e-12\\
1105	5.64892494768451e-12\\
1106	5.1925071600729e-12\\
1107	5.3881229807349e-12\\
1108	4.93556814081758e-12\\
1109	5.13811781580055e-12\\
1110	4.69084855998617e-12\\
1111	4.9010233838867e-12\\
1112	4.45744733555774e-12\\
1113	4.67533988919096e-12\\
1114	4.23500024146303e-12\\
1115	4.45859199494485e-12\\
1116	4.02089745412892e-12\\
1117	4.2519230167071e-12\\
1118	3.81733373960528e-12\\
1119	4.05383892227846e-12\\
1120	3.62209460560944e-12\\
1121	3.8650720526768e-12\\
1122	3.43605029818563e-12\\
1123	3.68556728050871e-12\\
1124	3.25825708306794e-12\\
1125	3.51480508191689e-12\\
1126	3.08959824238302e-12\\
1127	3.35088169553349e-12\\
1128	2.92697528945545e-12\\
1129	3.19480967902917e-12\\
1130	2.77217003559763e-12\\
1131	3.04580645481522e-12\\
1132	2.62476618864568e-12\\
1133	2.90406735267118e-12\\
1134	2.48392582602661e-12\\
1135	2.76928793220773e-12\\
1136	2.34907088022625e-12\\
1137	2.6400638225768e-12\\
1138	2.22146066639706e-12\\
1139	2.51761426087683e-12\\
1140	2.09820342781733e-12\\
1141	2.40079384226721e-12\\
1142	1.9817160619883e-12\\
1143	2.28888785700603e-12\\
1144	1.86893336938735e-12\\
1145	2.1830931928512e-12\\
1146	1.76202032338915e-12\\
1147	2.08223874248623e-12\\
1148	1.66018949267672e-12\\
};

\addplot [color=mycolor1, dashed, very thick]
  table[row sep=crcr]{%
0	1551.38172402755\\
1	655.780766698432\\
2	376.540483640398\\
3	267.518757593618\\
4	210.440579703027\\
5	175.023891144316\\
6	149.662292720963\\
8	115.875585090103\\
10	96.0054321069229\\
12	82.5304659748129\\
15	67.6312459560984\\
22	44.4949770396728\\
26	36.454039168392\\
32	28.181732683077\\
37	23.3307679485064\\
61	9.98089249103312\\
69	7.76264189477034\\
102	2.89554651712369\\
126	1.47115223861107\\
145	0.895497314596632\\
169	0.49023500809803\\
202	0.219714740977828\\
254	0.0638155931857063\\
312	0.016511860535878\\
392	0.0026229141403349\\
494	0.000257636550830719\\
621	1.46942684964587e-05\\
780	4.17652190540611e-07\\
979	4.97044827788881e-09\\
1200	3.72072066758924e-11\\
};

\addplot [color=mycolor2, very thick]
  table[row sep=crcr]{%
0	31.3225978495329\\
1	15.6001354691814\\
2	10.6222677120239\\
3	8.02644174486716\\
4	6.48375365140621\\
5	5.31720527386679\\
7	3.79761680510333\\
9	2.81613400131329\\
12	1.87173346378088\\
15	1.28036909000204\\
19	0.794749925461414\\
24	0.451150341545974\\
29	0.263179662470911\\
35	0.141703662881233\\
42	0.0706782255784262\\
50	0.0326857078996167\\
60	0.0127827876052058\\
74	0.00353254881763136\\
96	0.000484893391031939\\
132	1.96749690474131e-05\\
194	8.3592011265473e-08\\
291	1.72958435061624e-11\\
334	4.10396888867919e-13\\
344	1.74438821183751e-13\\
350	1.05912555844483e-13\\
351	9.84856940771201e-14\\
};

\addplot [color=mycolor2, dashed, very thick]
  table[row sep=crcr]{%
0	1551.38172402755\\
1	643.304113016299\\
2	348.336435329319\\
3	226.550549541636\\
4	164.134698705329\\
6	103.288084441524\\
8	71.4092525599431\\
10	51.6356959314961\\
13	33.478053609668\\
15	25.599434801853\\
17	19.7755342666206\\
20	13.719165659693\\
27	6.09075527367208\\
32	3.51848315314042\\
40	1.51029295109734\\
51	0.489686422358734\\
60	0.199263995154691\\
72	0.0615356775990491\\
83	0.0213241234085617\\
105	0.00266376620548964\\
143	7.95357393991391e-05\\
158	2.02754107779877e-05\\
167	8.94385704763661e-06\\
180	2.75922451670128e-06\\
191	1.0209696473836e-06\\
204	3.16840070232477e-07\\
216	1.07727835787719e-07\\
229	3.35224662137714e-08\\
246	7.31668579792965e-09\\
263	1.59953944061291e-09\\
290	1.43773750351525e-10\\
304	4.14699894688383e-11\\
318	1.21024908341567e-11\\
326	6.10386736354556e-12\\
327	5.54844112332544e-12\\
333	3.39475684534957e-12\\
336	2.86757942248792e-12\\
339	2.52082123284548e-12\\
343	1.91476582607644e-12\\
346	1.71386076539283e-12\\
348	1.63118451450811e-12\\
349	1.61315900086897e-12\\
350	1.62495151494964e-12\\
351	1.59347975197054e-12\\
352	1.6035357269207e-12\\
354	1.51207624337031e-12\\
355	1.45587049730419e-12\\
356	1.45165674559377e-12\\
357	1.41318410625803e-12\\
358	1.40694107638326e-12\\
360	1.44000735833051e-12\\
361	1.41407792415773e-12\\
362	1.33121667534659e-12\\
363	1.3216889512765e-12\\
364	1.25117962447471e-12\\
};

\end{axis}
\end{tikzpicture}%
    \caption{Simulation results of proximal algorithm \eqref{eq:gd_prox} with $\max(w)=1$ (blue) and $\max(w)=10$ (magenta).}
    \label{fig:simulation}
\end{figure}

\section{Epilogue}\label{sec:epilogue}

The link between algorithmic design in optimization and control synthesis provides a frame for tackling a host of questions on speed, accuracy, distributed computations, and so forth. It also offers a rich arsenal of control concepts and techniques that can be brought to bear. The scope of the present work has been to highlight basic results and insights that can be gained, as for instance, the frequency-domain recasting of optimization algorithms, the nature of analytic obstructions in performance as revealed by the gain margin theory, and how implicit algorithm may circumvent such.

\section*{Dedication}

During the past year, 2023, Allen Tannenbaum and Boris Polyak, founding pioneers of their respective fields of robust control theory and mathematical optimization, bode farewell to their subjects and science that they loved and contributed throughout their lives.
It would have been our greatest joy to have had the chance to share with them the insights gained in this past year that ultimately linked two of their defining contributions.
It is our hope that the present work that aims to highlight those links serves as a tribute and a celebration of their remarkable and lasting legacy.

\section*{Appendix~A: Proof of Proposition~\ref{thm:periodic}}\label{sec:proof_periodic}

Consider the configuration in Fig.~\ref{fig:gain_uncertain} where
$P_0(z)$ represents the lifted nominal plant and $C(z)$ a lifted $n$-periodic controller. Thus, both $P_0(z)$ and $C(z)$ are $n\times n$ transfer function matrices.
In order to ensure nominal stability, the complementary sensitivity $T(z)$ must satisfy the standard analytic constraints as follows.
 Assume that the plant has poles at $p_{i}$, $i=1,\dots,n_{p}$, $P_0(z)$ has poles at $p_{i}^n$, $i=1,\dots,n_{p}$. Therefore, $T(p_{i}^n)$ must have an eigenvalue equal to $1$, in that $T(p_{i}^n) v_{i} = v_{i}$ for a suitable nonzero $v_{i}\in\CI^n$, $i=1,\dots,n_{p}$.
Due to the causality of the controller, $C(\infty)$ must be lower triangular as explained earlier. Moreover, for similar reasons $P_0(\infty)$ is strictly lower triangular. Thus, $T(\infty)$ is also strictly lower triangular.

    Further, to guarantee robust stability for a range of values of an uncertain gain factor, as explained earlier for the case of LTI control, it is necessary and sufficient that $\det (I_n + (k-1)T(z))\neq 0$ for all $z\in\bD^c$ and $k\in[k_1,k_2]$. The permissible range of values of $\zeta\in\spec(T(z))$ (for $z\in\mathbb D^c$) should not intersect $\mathbb{S}_{k_1,k_2}$ defined in \eqref{eq:Sk1k2}. Using the same conformal map $\mathbf{\Phi}_{k_{1},k_{2}}$ in \eqref{eq:Phi}, the resulting $\mathbf{T}(z)$ as in \eqref{eq:bfT} must be analytic in $\bD^c$ with spectrum in $\bD$, and satisfy
\begin{subequations}\label{eq:intcond_mimo}
\begin{align}
&\mathbf{T}(p_{i}^n) v_{i} = \mathbf{\Phi}_{k_1,k_2}(1) v_{i} = g(k_1,k_2) v_{i} ,\; i=1,\dots,n_{p},\\
&\mathbf{T}(\infty) = \mathbf{\Phi}_{k_1,k_2}(T(\infty)) \text{ is strictly lower triangular}.
\end{align}
\end{subequations}
From discrete-time Lyapunov theory, $\spec(\mathbf{T}(z))\subset\bD$ if and only if there exists a positive definite matrix $R(z)$ such that
\begin{equation}\label{eq:spec_Ly}
    {\mathbf{T}(z)}\HP R(z) \mathbf{T}(z) - R(z) \prec \mathbf{0}.
\end{equation}
Consider the spectral factorization $R(z)={X(z)}\HP X(z)$. Then, \eqref{eq:spec_Ly} holds if and only if  $\overbar{\sigma}({X(z)} \mathbf{T}(z) X(z)^{-1}) < 1$ for all $z\in\bD^c$. Such a matrix function exists if and only if the Pick matrix \cite[Theorem 18.2.2]{ball_interpolation_1990}
\begin{equation*}
    \Lambda \coloneqq
\begin{bmatrix}
    \bigl(1-{g(k_1,k_2)}^2\bigr) \Xi_{1} &
    \Xi_{2}\HP  D_{g}\\
    D_{g}\HP  \Xi_{2} &
    I_n - D\HP  D
\end{bmatrix} \succ \mathbf{0},
\end{equation*}
where $D\coloneqq X(\infty)\mathbf{T}(\infty){X(\infty)}^{-1}$, $D_{g}\coloneqq I_n-g(k_1,k_2) D$, and, with $\xi_{i} \coloneqq X(p_{i}^n) v_{i}\in\CI^n$, $i=1,\dots,n_{p}$,
\begin{equation*}
    \Xi_{1} \coloneqq
    {\begin{bmatrix}
        \frac{\xi_{i}\HP \xi_{j}}{1- {\left( \overbar{p_{i}} p_{j} \right)}^{-n}}
    \end{bmatrix}}_{1 \leq i,j \leq n_{p}} , \quad
    \Xi_{2} \coloneqq {\left[ \xi_{1}  \;\cdots\;  \xi_{n_{p}} \right]}.
\end{equation*}
We claim that $\Xi_{1}\succ \mathbf{0}$.
To see this, consider any nonzero $a\in\CI^{n_{p}}$. Then,
\begin{align*}
    a\HP \Xi_{1} a &= \sum_{i=1}^{n_{p}} \sum_{j=1}^{n_{p}} \frac{\xi_{i}\HP \xi_{j} \overbar{a_{i}} a_{j} }{1- {( \overbar{p_{i}} p_{j} )}^{-n}} \\
    &= \sum_{k=0}^{\infty} \sum_{i=1}^{n_{p}} \sum_{j=1}^{n_{p}} \xi_{i}\HP \xi_{j} \overbar{a_{i}} a_{j} {( \overbar{p_{i}} p_{j} )}^{-nk} \\
    &= \sum_{k=0}^{\infty} {\left( \sum_{i=1}^{n_{p}} \frac{a_{i}}{p_{i}^{nk}} \xi_{i} \right)\!\!}\HP {\left( \sum_{i=1}^{n_{p}} \frac{a_{i}}{p_{i}^{nk}} \xi_{i} \right)\!} > 0.
\end{align*}
Thus, $\Xi_{1}\succ \mathbf{0}$, which, together with the Schur complement, equates $\Lambda \succ \mathbf{0}$ to
\begin{equation}\label{eq:pick2}
    {\bigl( 1- {g(k_1,k_2)}^2 \bigr)} \bigl(I_n - D\HP  D\bigr) \succ
    D_{g}\HP \Xi_{2} \Xi_{1}^{-1} \Xi_{2}\HP  D_{g}
\end{equation}
A direct calculation shows that
\begin{equation}\label{eq:1g}
    \bigl(1-{g(k_1,k_2)}^2\bigr) \bigl(I_n - {{D}}\HP  {D} \bigr) =
    D_{g}\HP {\bigl( I_{n} - {\widetilde{D}}\HP  \widetilde{D} \bigr)} D_{g} ,
\end{equation}
for $\widetilde{D}=\bigl(g(k_1,k_2) I_n- {{D}}\bigr) D_{g}^{-1}$, where the invertibility of $D_{g}$ follows from $\det \bigl(D_{g}\bigr) = \det \bigl(I_n-g(k_1,k_2) \mathbf{T}(\infty)\bigr) = 1$ as $g(k_1,k_2)$ is a scalar and $D$ is similar to the strictly lower diagonal matrix $\mathbf{T}(\infty)$. Due to the same reason, we have
\begin{equation}\label{eq:detD}
    \det (\widetilde{D}) = \frac{\det(g(k_1,k_2) I_n- \mathbf{T}(\infty))}{\det( D_{g})} ={g(k_1,k_2)}^{n}.
\end{equation}
Using the identity \eqref{eq:1g}, the inequality \eqref{eq:pick2} is equivalent to
$I_{n} - \Xi_{2} \Xi_{1}^{-1} \Xi_{2}\HP \succ {\widetilde{D}}\HP  \widetilde{D}$,
for which a necessary condition is
\begin{equation}\label{eq:det1}
    \det \bigl(I_{n} - \Xi_{2} \Xi_{1}^{-1} \Xi_{2}\HP \bigr) > {\bigl|\det (\widetilde{D})\bigr|}^2.
\end{equation}
Invoking the matrix determinant lemma and the identity
\begin{equation*}
    \Xi_{1} - \Xi_{2}\HP  \Xi_{2} = {\diag (p_{1}^{-n},\dots,p_{n_p}^{-n})}\HP  \Xi_{1} \diag (p_{1}^{-n},\dots,p_{n_p}^{-n}),
\end{equation*}
we obtain
\begin{equation*}
    \det (I_{n} - \Xi_{2} \Xi_{1}^{-1} \Xi_{2}\HP ) = \frac{\det (\Xi_{1} - \Xi_{2}\HP  \Xi_{2} )}{\det (\Xi_{1})} = \prod_{i=1}^{n_{p}} {\left| p_{i} \right|}^{-2n}.
\end{equation*}
Substituting this and \eqref{eq:detD} into \eqref{eq:det1} yields
\begin{equation*}
    \prod_{i=1}^{n_{p}} {\left| p_{i} \right|}^{-2n} > {g(k_1,k_2)}^{2n},
\end{equation*}
which is equivalent to \eqref{eq:unmistable} and thus completes the proof.



\section*{Appendix~B: Proof of Theorem~\ref{thm:circle_prox}}\label{sec:proof_prox}

Since $\eta_{1}$ and $\eta_{2}$ are free to tune, the condition about $\Psi(\infty)$ in \eqref{eq:Psi_inter} is equal to the tangential interpolation condition $\Psi(\infty) \xi = \xi$ with $\xi = {\left[ 1 \; 0 \right]}\TP$. Let
\begin{equation*}
    \setlength{\arraycolsep}{0.7ex}
    P(z) \!=\! \begin{bmatrix}
    w & 0 \\ 0 & 1
\end{bmatrix} \! \Psi(z) \!
\begin{bmatrix}
    {1}/{w} & 0 \\ 0 & 1
\end{bmatrix} , \quad
P_{1} \!=\! \begin{bmatrix}
    w & 0 \\ 0 & 1
\end{bmatrix} \! \Psi(1) \!
\begin{bmatrix}
    {1}/{w} & 0 \\ 0 & 1
\end{bmatrix}.
\end{equation*}
Then $P(z)$ must satisfy the interpolation conditions $P(1) = P_{1}$ and $P(\infty) \xi = \xi$, and $P(\gamma z)$ is strictly positive real for all $\gamma\in(\rho,1)$.
From the interpolation theory for positive real functions, the existence of such a function $P(z)$ is equivalent to the positive definiteness of the Caratheodory--Pick matrix \cite[Theorem 22.2.1]{ball_interpolation_1990}
\begin{equation*}
    \Lambda =
    \begin{bmatrix}
        \frac{1}{1-\gamma^2} \bigl( P_1 + P_1\TP \bigr) &
        \xi + P_1\TP \xi \\
        \xi\TP + \xi\TP P_1  & 2 \xi\TP \xi
    \end{bmatrix}\quad \forall \gamma\in(\rho,1).
\end{equation*}
We apply a congruence transformation to $\Lambda$:
\begin{equation*}
    \begin{bmatrix}
        1 & 0 & 0 \\ 0 & 0 & 1 \\ 0 & 1 & 0
    \end{bmatrix}
    \Lambda
    \begin{bmatrix}
        1 & 0 & 0 \\ 0 & 0 & 1 \\ 0 & 1 & 0
    \end{bmatrix}
    = \begin{bmatrix}
        \Lambda_{11} & \Lambda_{12} \\ \Lambda_{12}\TP & \Lambda_{22}
    \end{bmatrix},
\end{equation*}
where
\begin{equation*}
    \Lambda_{11} = \begin{bmatrix}
        \frac{\ell_{1} + \mu_{2}}{\mu_{1}+\mu_{2}} \frac{2}{1-\gamma^2} & \frac{\ell_{1}+\mu_{2}}{\mu_{1}+\mu_{2}} + 1 \\
        \frac{\ell_{1}+\mu_{2}}{\mu_{1}+\mu_{2}} + 1 & 2
    \end{bmatrix}, \quad
    \Lambda_{12} = \begin{bmatrix}
        \frac{(\ell_{1}-\mu_{1}) w + (\ell_{2}-\mu_{2}) / w}{(\mu_{1}+\mu_{2}) (1-\gamma^2)} \\
        \frac{\ell_{1}-\mu_{1}}{\mu_{1} + \mu_{2}} w
    \end{bmatrix}, \quad
    \Lambda_{22} =
        \frac{\ell_{2} + \mu_{1}}{\mu_{1} + \mu_{2}}\frac{2}{1-\gamma^2}.
\end{equation*}
Necessary conditions for $\Lambda \succ \mathbf{0}$ include $\Lambda_{11} \succ \mathbf{0}$ and $\Lambda_{22} > 0$, which amount to $\rho_{\mathrm{S}}^2 < \gamma^2 < 1$, where $\rho_{\mathrm{S}}$ is defined in \eqref{eq:rho_s}.
From the construction of $\Lambda$, to make $\Lambda \succ \mathbf{0}$ for all $\gamma\in(\rho_{\mathrm{S}},1)$ it suffices that $\Lambda\succeq \mathbf{0}$ at $\gamma=\rho_{\mathrm{S}}$. When $\gamma = \rho_{\mathrm{S}}$, $\Lambda_{11}$ can be factored as $\Lambda_{11} = 2 \nu \nu\TP$ where
\begin{equation*}
    \nu =  {\begin{bmatrix}
        {\left( 1-\rho_{\mathrm{S}} \right)}^{-1} & 1
    \end{bmatrix}}\TP.
\end{equation*}
Since $\Lambda_{{22}}>0$ when $\gamma = \rho_{\mathrm{S}}$, using the Schur complement, $\Lambda\succeq \mathbf{0}$ is equivalent to
\begin{equation}\label{eq:split_pick2}
    \Lambda_{11} -  \Lambda_{12} \Lambda_{22}^{-1} \Lambda_{21}\TP  =
    2\nu \nu\TP - \Lambda_{12} \Lambda_{22}^{-1} \Lambda_{21}\TP \succeq \mathbf{0}.
\end{equation}
A necessary condition for this is that $\Lambda_{12}=k \nu$ for some scalar $k$. It is equivalent to
\begin{equation*}
    \frac{(\ell_{1}-\mu_{1}) w + (\ell_{2}-\mu_{2}) / w}{(\mu_{1}+\mu_{2}) (1-\rho_{\mathrm{S}}^2)} = \frac{\ell_{1}-\mu_{1}}{\mu_{1}+\mu_{2}} \frac{w}{1-\rho_{\mathrm{S}}},
\end{equation*}
which leads to $w^2 = \frac{\ell_{2}-\mu_{2}}{\ell_{1}-\mu_{1}} \rho_{\mathrm{S}}^{-1}$ and $k^{2}=\frac{(\ell_{1}-\mu_{1}) (\ell_{2}-\mu_{2})}{{( \mu_{1}+\mu_{2} )}^{2} \rho_{\mathrm{S}}}$. In this case, \eqref{eq:split_pick2} holds because of
\begin{equation*}
    2 - \frac{k^{2}}{\Lambda_{22}} = \frac{2(\ell_{1}+\ell_{2}+\mu_{1}+\mu_{2}) (\mu_{1}+\mu_{2})}{(\ell_{2}+\mu_{1}) (\ell_{1}+\mu_{1}+2\mu_{2})} > 0.
\end{equation*}
We now dive into the construction of the function $\Psi(z)$ for $\rho = \rho_{\mathrm{S}}$.
Using the interpolation theory for positive real functions \cite[Theorem 22.2.1]{ball_interpolation_1990}, we obtain a satisfactory function $P(z)$ as
\begin{equation*}
    P(z) = \frac{1}{z - \rho_{\mathrm{S}}}
    \begin{bmatrix}
            z + \rho_{\mathrm{S}} &
            2 \sqrt{\frac{\ell_2-\mu_2}{\ell_1-\mu_1}\rho_{\mathrm{S}}}  z \\
            \frac{2 \sqrt{(\ell_1-\mu_1) (\ell_2-\mu_2) \rho_{\mathrm{S}}} }{\ell_{1}+\mu_{1} + 2\mu_{2}} &
            \frac{\ell_1+\mu_1+2\ell_2}{\ell_{1}+\mu_{1} + 2\mu_{2}} z - \rho_{\mathrm{S}}
        \end{bmatrix}.
\end{equation*}
The proof is complete by $\Psi(z) = \begin{bsmallmatrix}
    1/w & 0 \\ 0 & 1
\end{bsmallmatrix} P(z) \begin{bsmallmatrix}
    {w} & 0 \\ 0 & 1
\end{bsmallmatrix}$.

\section*{Acknowledgment}
We are grateful to Patrizio Colaneri for insightful discussions on the subject of this work.

\bibliographystyle{IEEEtran}
\bibliography{optimization}

\end{document}